\documentclass[sigconf]{acmart}

\AtBeginDocument{%
  \providecommand\BibTeX{{%
    \normalfont B\kern-0.5em{\scshape i\kern-0.25em b}\kern-0.8em\TeX}}}

\renewcommand\footnotetextcopyrightpermission[1]{} 
\setcopyright{none}
\acmDOI{}

\usepackage{graphicx}
\usepackage{footnote} 
\usepackage{balance}  
\usepackage{amsmath,amsfonts} 
\usepackage{bm}
\usepackage{algorithmic}
\usepackage{graphicx}
\usepackage{textcomp}
\usepackage[labelformat=empty,justification=centering]{subfig} 
\usepackage{soul,color}
\usepackage{xspace}
\usepackage{listings}
\usepackage{verbatimbox}
\usepackage{url}
\usepackage{stmaryrd}
\usepackage{soul,color}
\usepackage{xcolor}
\definecolor{codegreen}{rgb}{0,0.6,0}

\newcommand{\rdfframes}{RDFFrames\xspace}
\newcommand{\rdfframe}{RDFFrame\xspace}
\newcommand{\sparql}{SPARQL\xspace}
\newcommand{\squishlist}{
  \begin{list}{$\bullet$}
   {
     \setlength{\itemsep}{0pt}
     \setlength{\parsep}{0pt}
     \setlength{\topsep}{0pt}
     \setlength{\partopsep}{0pt}
     \setlength{\leftmargin}{1.5em}
     \setlength{\labelwidth}{1em}
     \setlength{\labelsep}{0.5em} } }
\newcommand{\squishend}{
   \end{list}  }

\newcommand{\hide}[1]{}

\newtheorem{definition}{Definition}
\newtheorem{lemma}{Lemma}
\newtheorem{theorem}{Theorem}

\DeclareFontEncoding{LS1}{}{}
\makeatletter

\DeclareFontSubstitution{LS1}{stix}{m}{n}

\DeclareSymbolFont{symbols2}      {LS1}{stixfrak} {m} {n}

\DeclareMathSymbol{\leftouterjoin}            {\mathop}   {symbols2}{"11}
\DeclareMathSymbol{\rightouterjoin}           {\mathop}   {symbols2}{"12}
\DeclareMathSymbol{\fullouterjoin}            {\mathop}   {symbols2}{"13}
\DeclareMathSymbol{\innerjoin}                     {\mathop}   {symbols2}{"ED}

\newenvironment{packed_enum}{
\begin{enumerate}
  \setlength{\itemsep}{1pt}
  \setlength{\parskip}{0pt}
  \setlength{\parsep}{0pt}
}{\end{enumerate}}

\makeatletter
\renewcommand\@formatdoi[1]{\ignorespaces}
\makeatother

\begin{document}

\title{\rdfframes: Knowledge Graph Access for Machine Learning Tools}

\author{
  Aisha Mohamed{$^{1*}$}
  \mbox{     }
  Ghadeer Abuoda{$^{1\mathsection}$}
  \mbox{     }
  Abdurrahman Ghanem{$^{2\dagger}$} 
  \and Zoi Kaoudi{$^{2\ddagger}$} 
  \mbox{     }
  Ashraf Aboulnaga{$^{*}$}
}
\affiliation{
  \institution{$^{*}$ Qatar Computing Research Institute, HBKU}
  \institution{$^{\mathsection}$ College of Science and Engineering, HBKU}
  \institution{$\dagger$Bluescape}
  \institution{$^{\ddagger}$Technische Universit\"{a}t Berlin}
}\thanks{$^{1}$ Joint first authors.}
\thanks{$^{2}$ Work done while at QCRI.}
  \email{{ahmohamed, gabuoda, aaboulnaga}@hbku.edu.qa}
  \email{ghanemabdo@gmail.com, zoi.kaoudi@tu-berlin.de}
  

\begin{abstract}
Knowledge graphs represented as RDF datasets are integral to many machine learning applications. RDF is supported by a rich ecosystem of data management systems and tools, most notably RDF database systems that provide a SPARQL query interface. Surprisingly, machine learning tools for knowledge graphs do not use SPARQL, despite the obvious advantages of using a database system. This is due to the mismatch between SPARQL and machine learning tools in terms of data model and programming style. Machine learning tools work on data in tabular format and process it using an imperative programming style, while SPARQL is declarative and has as its basic operation matching graph patterns to RDF triples. We posit that a good interface to knowledge graphs from a machine learning software stack should use an imperative, navigational programming paradigm based on graph traversal rather than the SPARQL query paradigm based on graph patterns. In this paper, we present RDFFrames, a framework that provides such an interface. RDFFrames provides an imperative Python API that gets internally translated to SPARQL, and it is integrated with the PyData machine learning software stack. RDFFrames enables the user to make a sequence of Python calls to define the data to be extracted from a knowledge graph stored in an RDF database system, and it translates these calls into a compact SPQARL query, executes it on the database system, and returns the results in a standard tabular format. Thus, RDFFrames is a useful tool for data preparation that combines the usability of PyData with the flexibility and performance of RDF database systems.
\end{abstract}

\settopmatter{printfolios=true}
\maketitle
\pagestyle{plain} 

\section{Introduction}
\label{sec:intro}

There has recently been a sharp growth in the number of knowledge graph datasets that are made available in the RDF (Resource Description Framework)\footnote{\url{https://www.w3.org/RDF}} data model.
Examples include knowledge graphs that cover a broad set of domains such as DBpedia~\cite{lehmann2015dbpedia}, YAGO~\cite{yago2}, Wikidata~\cite{vrandecic12wikidata},
and BabelNet~\cite{Navigli:2012:BAC:2397213.2397579}, 
as well as specialized graphs for specific domains like 
product graphs for e-commerce~\cite{Dong:2018:CIB:3219819.3219938}, biomedical information networks~\cite{belleau2008bio2rdf}, 
and bibliographic datasets~\cite{mccallum2000automating, giles1998citeseer}.
The rich information and semantic structure of knowledge graphs makes them 
useful in many machine learning applications~\cite{davis1993knowledge},
such as recommender systems~\cite{jenatton2012latent},
virtual assistants, and question answering systems~\cite{west2014knowledge}.
Recently, many machine learning algorithms have been developed specifically for 
knowledge graphs,
especially in the sub-field of \textit{relational learning}, which is dedicated to learning from
the 
relations between entities in a knowledge graph~\cite{nickel2015review, nguyen2017overview, wang2018multi}.

RDF is widely used to publish knowledge graphs as it provides a powerful abstraction for representing heterogeneous, incomplete, sparse, and potentially noisy knowledge graphs. RDF is supported by a rich ecosystem of data management systems and tools that has evolved over the years.
This ecosystem includes standard serialization formats, parsing and processing libraries,
and most notably RDF database management systems (a.k.a. \textit{RDF engines} or \textit{triple stores}) that support \sparql,\footnote{\url{https://www.w3.org/TR/sparql11-query}}
the W3C standard query language for RDF data.
Examples of these systems include
OpenLink Virtuoso,\footnote{\url{https://virtuoso.openlinksw.com}}
Apache Jena,\footnote{\url{https://jena.apache.org}}
and
managed services such as
Amazon Neptune.\footnote{\url{https://aws.amazon.com/neptune}}

\textit{However, we make the observation that none of the publicly available machine learning or relational learning tools for knowledge graphs that we are aware of uses \sparql
to explore and extract datasets from knowledge graphs stored in RDF database systems.} 
 This, despite the obvious advantage of using a database system such as data independence, declarative querying, and efficient and scalable query processing.
For example, we investigated all the prominent recent open source relational learning implementations,
and we found that they all rely on ad-hoc scripts to process very small knowledge graphs and prepare the necessary datasets for learning.
This observation applies to the implementations of published state-of-the-art embedding models, e.g., scikit-kge~\cite{nickel2016holographic, nickel2011three},\footnote{\url{https://github.com/mnick/scikit-kge}}
and also holds for the recent Python libraries that are currently used as standard implementations for training and benchmarking knowledge graph embeddings, e.g.,
Ampligraph~\cite{ampligraph}, OpenKE~\cite{han2018openke}, and PyKEEN~\cite{ali2020pykeen}.
These scripts are limited in performance, which slows down data preparation and leaves the challenges of applying embedding models on the scale of real knowledge graphs unexplored. 

We posit that machine learning tools do not use RDF engines due to an ``impedance mismatch.''
Specifically, typical machine learning software stacks are based on data in \textit{tabular format} and the split-apply-combine paradigm~\cite{Wickham_2011}.
An example tabular format is the highly popular \textit{dataframes},
supported by libraries in several languages such as Python and R
(e.g., the pandas\footnote{\url{https://pandas.pydata.org}}
and scikit-learn libraries in Python),
and by systems such as Apache Spark~\cite{apacheSpark}.
Thus, the first step in most machine learning pipelines
(including relational learning) 
is a data preparation step that
explores the knowledge graph,
identifies the required data,
extracts this data from the graph,
efficiently processes and cleans the extracted data,
and returns it in a table.
Identifying and extracting this refined data from a knowledge graph requires efficient and flexible \textit{graph traversal} functionality. \sparql is a declarative pattern matching query language
designed for distributed data integration with unique identifiers 
rather than navigation~\cite{Matsumoto2018MappingRG}. 
Hence, while \sparql has the expressive power to process and extract data into tables, machine learning tools do not use it  since it lacks the required flexibility and ease of use of \textit{navigational} interfaces.

In this paper, we introduce \textit{\rdfframes},
a framework that bridges the gap between machine learning tools and RDF engines.
\rdfframes is designed to support the data preparation step.
It defines a user API consisting of two type of operators: navigational operators that explore an RDF graph and extract data from it based on a graph traversal paradigm, and relational operators for processing this data into refined clean datasets for machine learning applications.
The sequence of operators called by the user represents a logical description of the required dataset.
\rdfframes translates this description to a corresponding \sparql query, executes it 
on an RDF engine,
and returns the results as a table.

In principle, the \rdfframes operators can be implemented in any programming language and can return data in any tabular format.
However, concretely, our current implementation of \rdfframes is a Python library that returns data as dataframes of the popular pandas 
library
so that further processing can leverage
the richness
of the PyData ecosystem.
\rdfframes is available as open source\footnote{\url{https://github.com/qcri/rdfframes}}
and via the Python 
pip installer. It is implemented in 6,525 lines of code, and was demonstrated in~\cite{rdfframes-demo}.

\vspace{4pt}
\noindent
\textbf{Motivating Example.}
We illustrate the end-to-end operation of \rdfframes through an example.
Assume the DBpedia knowledge graph is stored in an RDF engine,
and consider a machine learning practitioner who wants 
use DBpedia
to study
prolific American actors (defined as those who have starred in 50 or more movies).
Let us say that the practitioner wants to see
the movies these actors starred in
and the Academy Awards won by any of them.
Listing~\ref{lst:motivating_code} shows Python code using the \rdfframes API that prepares a dataframe with the
data required for this task.
It is important to note that this code is a \textit{logical description} of the dataframe and \textit{does not cause a query to be generated or data to be retrieved from the RDF engine.}
At the end of a sequence of calls such as these, the user calls a special \texttt{execute}
function that causes a \sparql query to be generated and executed on the engine, and the results to be returned in a dataframe.

The first statement of the code creates a two-column
\rdfframe with 
the URIs (Universal Resource Identifiers) of
all movies and all the actors who starred in them.
The second statement navigates from the actor column in this 
\rdfframe to get the birth place of each actor and uses a filter
to keep only American actors.
Next, the code finds all American actors who have starred
in 50 or more movies (prolific actors).
This requires grouping and aggregation, as well as a filter 
on the aggregate value.
The final step is to navigate from the actor column in the prolific actors 
\rdfframe to get the actor's Academy Awards (if available).
The result dataframe will contain the prolific actors, movies that they starred in, and their Academy Awards if available.
An expert-written \sparql query corresponding to Listing~\ref{lst:motivating_code} is shown in Listing~\ref{lst:motivating_query}.
\rdfframes provides an alternative to writing such a \sparql
query that is simpler and closer to the navigational paradigm and is
better-integrated with the machine learning environment.
The case studies in Section~\ref{subsec:case-studies} describe 
more complex data preparation tasks and present the
\rdfframes code for these tasks and the corresponding
\sparql queries.

\vspace{8pt}
\begin{lstlisting}[
  aboveskip=-0.0\baselineskip,
  belowskip=-0.0\baselineskip,
  language=Python,
  breaklines=true,
  showspaces=false,
  basicstyle=\ttfamily\scriptsize,
  commentstyle=\color{gray},
  otherkeywords={expand, filter, group_by, join, feature_domain_range, .count, cache, unique},
  keywordstyle=\color{blue},
  caption={\rdfframes code - Prolific American actors who have Academy Awards.},
  captionpos=b,
  label={lst:motivating_code}]
movies = graph.feature_domain_range('dbp:starring',
  'movie', 'actor')
american = movies.expand('actor',
  [('dbp:birthPlace', 'country')])\
  .filter({'country': ['=dbpr:United_States']})
prolific = american.group_by(['actor'])\
  .count('movie', 'movie_count')\
  .filter({'movie_count': ['>=50']})
result = prolific.expand('actor', [('dbpp:starring',
  'movie', INCOMING), ('dbpp:academyAward', 'award',
   OPTIONAL)])
\end{lstlisting}

\vspace{8pt}
\begin{lstlisting}[
aboveskip=-0.0\baselineskip,
belowskip=-0.0\baselineskip,
language=SQL,
breaklines=true,
showspaces=false,
basicstyle=\ttfamily\scriptsize,
commentstyle=\color{gray},
lineskip=-0.7ex,
breakatwhitespace=true,
keywordstyle=\color{blue},
otherkeywords={OPTIONAL, FILTER, UNION, SELECT, GROUP BY ,WHERE, FROM },
caption={Expert-written \sparql query corresponding to \rdfframes code shown in Listing~\ref{lst:motivating_code}.},
label={lst:motivating_query},
captionpos=b,
]

SELECT  *
FROM <http://dbpedia.org>
WHERE
  { ?movie  dbpp:starring  ?actor
    { SELECT DISTINCT  ?actor
        (COUNT(DISTINCT ?movie) AS ?movie_count)
      WHERE
        { ?movie  dbpp:starring    ?actor .
          ?actor  dbpp:birthPlace  ?actor_country
          FILTER ( ?actor_country = dbpr:United_States )
        }
      GROUP BY ?actor
      HAVING ( COUNT(DISTINCT ?movie) >= 50 )
    }
    OPTIONAL
      { ?actor  dbpp:academyAward  ?award }
  }
\end{lstlisting}

\begin{figure}
  \centering
  \includegraphics[width=0.8 \columnwidth]{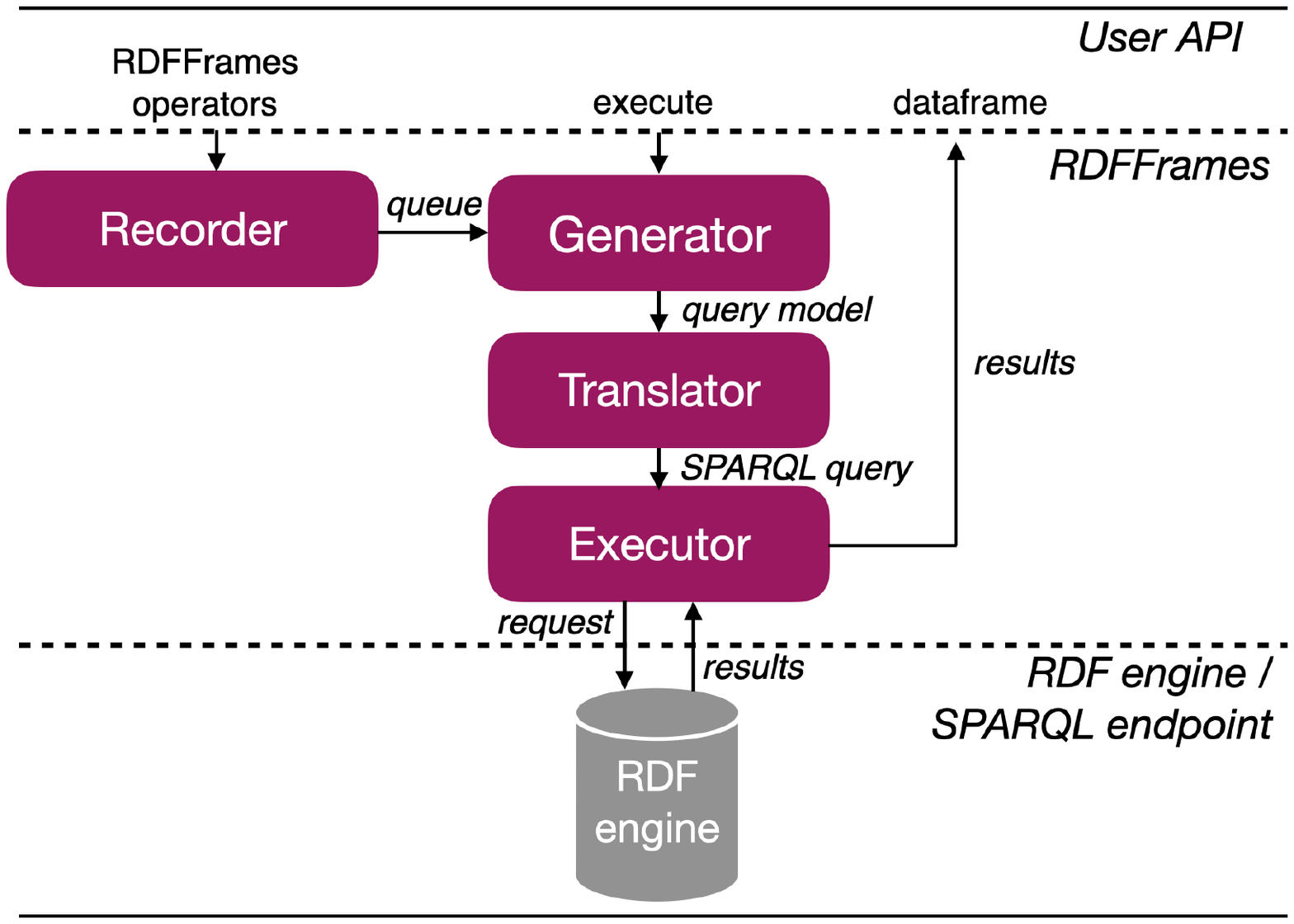}
  \caption{\rdfframes architecture.}
  \label{fig:architecture}
  \vspace*{-12pt}
\end{figure}

\noindent
\textbf{\rdfframes in a Nutshell.}
The architecture of \rdfframes is shown in Figure~\ref{fig:architecture}.
At the top of the figure is the user API, which consists of a set of operators implemented as Python functions.
We make a design decision in \rdfframes to use a \textit{lazy evaluation strategy}.
Thus, the \texttt{Recorder} records the operators invoked by the user without executing them,
storing the operators in a FIFO queue.
The special \texttt{execute} operator causes the 
\texttt{Generator} to consume the operators in the queue and build a {\em query model} representing the user's code.
The query model is an intermediate representation for
\sparql queries. The goal of the query model is (i)~to separate the API parsing logic from the query building logic for flexible manipulation and implementation, and (ii)~to facilitate optimization techniques for building the queries, especially in the case of nested queries.
Next, the \texttt{Translator} translates the query model into a \sparql query.
This process includes validation to ensure that the generated query has valid \sparql syntax and is equivalent to the user's API calls.
Our choice to use lazy evaluation means that the entire sequence of operators called by the user is captured in the query model processed by the \texttt{Translator}.
We design the \texttt{Translator} to take advantage of this fact and generate
\textit{compact and efficient} \sparql queries.
Specifically, each query model is translated to one \sparql query and the \texttt{Translator} minimizes the number of nested subqueries.
After the \texttt{Translator}, the \texttt{Executor} sends the
generated \sparql
query to an RDF engine or \sparql endpoint, handles all communication issues,
and returns the results to the user in a dataframe.

\noindent

\textbf{Contributions:}
The novelty of \rdfframes lies in:
\squishlist
\item
First, the API provided to the user is designed to be intuitive and flexible,
in addition to being expressive.
The API consists of navigational operators and data processing operators based on familiar relational algebra operations
such as filtering, grouping, and joins (Section~\ref{sec:data-model}).

\item
Second, \rdfframes translates the API calls into efficient \sparql queries.
A key element in this is the query model which exposes query equivalences in a simple way.
In generating the query model from a sequence of API calls and in generating the \sparql query from the query model,
\rdfframes has the overarching goal of generating efficient queries
(Section~\ref{sec:query-generation}).
We prove the correctness of the translation from API calls to \sparql. That is, we prove that the dataframe that \rdfframes returns is semantically equivalent to the results set of the generated \sparql query (Section~\ref{sec:correctness}).

\item
Third, \rdfframes handles all the mechanics of processing the \sparql query such as the connection to the RDF engine 
or \sparql endpoint, pagination (i.e., retrieving the results in chunks) to avoid the endpoint timing out, and converting the result to a dataframe. 
We present case studies and performance comparisons that validate our design decisions
and
show that \rdfframes outperforms several
alternatives
(Section~\ref{sec:performance}).

\squishend

\section{Related Work}
\label{sec:related}

\noindent
\textbf{Data Preparation for Machine Learning.}
It has been reported that 80\% of data analysis time and effort is spent on
the process of exploring, cleaning, and preparing the data~\cite{dasu2003exploratory},
and these activities have long been a focus of the database community.
For example, the recent Seattle report on database research~\cite{SeattleDBReport2019}
acknowledges the importance of these activities and the need to support data science ecosystems such as PyData and to ``devote more efforts on the end-to-end data-to-insights pipeline.''
This paper attempts to reduce the data preparation effort by defining a powerful API for accessing knowledge graphs.
To underscore the importance of such an API, note that~\cite{hidden_debt} makes the observation that most of the code in a machine learning system is devoted to tasks other than learning and prediction.
These tasks include collecting and verifying data and preparing it for use in machine learning packages.
This requires a massive amount of ``glue code'', and~\cite{hidden_debt} observes that this glue code can be eliminated by using well-defined common APIs for data access (such as \rdfframes).

Some related work focuses on the end-to-end  machine learning life cycle (e.g.,~\cite{tfx,mlflow,mldp}).
Some systems, such as MLdp~\cite{mldp}, focus primarily on managing input data, 
but they do not have special support for knowledge graphs. \rdfframes provides such support.

\vspace{4pt}
\noindent
\textbf{Database Support for PyData.}
Some recent efforts to provide database support for the PyData ecosystem focus on the scalability of dataframe operations, while other efforts focus on replacing SQL as the traditional data access API with pandas-like APIs.
Koalas\footnote{\url{https://koalas.readthedocs.io/en/latest}}
implements the pandas dataframe API on top of Apache Spark for better scalability.
Modin~\cite{modin} is a scalable dataframe system based on a novel formalism for pandas dataframes.
Ibis\footnote{\url{https://ibis-project.org}}
defines a variant of the dataframe API
(not pandas)
and translates it to SQL so that it can execute on a database system for scalability. Ibis also supports other backends such as Spark.
Koalas and Modin do not support SQL backends, and Ibis does not have a pandas API.
A recent system that addresses these limitations is
Magpie~\cite{magpie},
which
translates pandas operations to Ibis for scalable execution on multiple backends, including SQL database systems. Magpie chooses the best backend for a given program based on the program's complexity and the data size.
Grizzly~\cite{grizzly} is a framework that generates SQL queries from a pandas-like API and ships the SQL to a standard database system for scalable execution.
Grizzly relies on the database system's support for external tables in order to load the data.
It also 
creates user UDFs as native UDFs in the database system.
The AIDA framework~\cite{AIDA2018} allows users to write relational and linear algebra operators in Python and pushes the execution of these operators into a relational database system.

All of these recent works are similar in spirit to
\rdfframes in that they replace SQL for data access with a pandas-like API and/or rely on a database backend for scalability. However, all of the works focus on relational data and not graph data.
\rdfframes is the first to define a pandas-like API for graph data (specifically RDF),
and to support a graph database system as a scalable backend.

\vspace{4pt}
\noindent
\textbf{Why RDF?}
Knowledge graphs are typically represented in the RDF data model.
Another popular data model for graphs is the \textit{property graph} data model,
which has labels on nodes and edges as well as (property, value) pairs associated with both.
Property graphs have gained wide adoption in many applications and are supported by popular database systems such as Neo4j\footnote{\url{https://neo4j.com}}
and Amazon Neptune.
Multiple query languages exist for property graphs, and efforts are underway to define a common powerful query language~\cite{g-core}.

A popular query language for property graphs is Gremlin.\footnote{\url{https://tinkerpop.apache.org/gremlin.html}}
Like \rdfframes, Gremlin adopts a navigational approach to querying the graph,
and some of the \rdfframes operators are similar to Gremlin operators.
The popularity of Gremlin is evidence that a navigational approach is attractive to users.
However, all publicly available knowledge graphs including DBpedia~\cite{lehmann2015dbpedia} and YAGO~\cite{yago} are represented in RDF format. Converting RDF graphs to property graphs is not straightforward mainly because the property graph model does not provide globally unique identifiers and linking capability as a basic construct. 
In RDF knowledge graphs, each entity and relation is uniquely identified by a URI, and links between graphs are created by using the URIs from one graph in the other.
\rdfframes offers a navigational API similar to Gremlin to data scientists working with knowledge graphs in RDF format and facilitates the integration of this API with the data analysis tools of the PyData ecosystem. 

\vspace{4pt}
\noindent
\textbf{Why \sparql?}
\rdfframes uses \sparql as the interface for accessing knowledge graphs.
In the early days of RDF, several other query languages were proposed
(see~\cite{haase2004rdfquerylanguages} for a survey), but none of them has seen broad adoption, and \sparql has emerged as the standard.

Some work proposes navigational extensions to \sparql (e.g.,~\cite{nsparql, sparqler}), but these proposals add complex navigational constructs such as path variables
and regular path expressions to the language.
In contrast, the navigation used in \rdfframes is simple and well-supported by standard \sparql
without extensions.
The goal of \rdfframes is not complex navigation, but rather providing a simple yet common and rich suite of data access and preparation operators that can be integrated in a machine learning pipeline.

\vspace{4pt}
\noindent
\textbf{Python Interfaces.}
A Python interface for accessing RDF knowledge graphs is provided by
Google's Data Commons project.\footnote{\url{http://datacommons.org}}
However, the goal of that project is not to provide powerful data access, but rather to synthesize a single graph from multiple knowledge graphs, and to enable browsing for graph exploration.
The provided Python interface has only one data access primitive: following an edge in the graph in either direction, which is but one of many capabilities provided by \rdfframes.


The Magellan project~\cite{magellan} provides a set of interoperable Python tools for entity matching pipelines. It is another example of developing data management solutions by extending the PyData ecosystem~\cite{anhai_hilda}, albeit in a very different domain from \rdfframes. The same factors that made Magellan successful in the world of entity matching can make \rdfframes successful in the world of knowledge graphs. 

There are multiple recent Python libraries that provide access to knowledge graphs through a \sparql endpoint over HTTP.
Examples include
pysparql,\footnote{\url{https://code.google.com/archive/p/pysparql}}
sparql-client,\footnote{\url{https://pypi.org/project/sparql-client}}
and 
AllegroGraph Python client.\footnote{\url{https://franz.com/agraph/support/documentation/current/python}}
However, all these libraries solve a very different (and simpler) problem compared to \rdfframes: they take a \sparql query written by the user and handle sending this query to the endpoint and receiving results. On the other hand, the main focus of \rdfframes is \textit{generating} the \sparql query from imperative API calls. Communicating with the endpoint is also handled by \rdfframes, but it is not the primary contribution.

\vspace{4pt}
\noindent
\textbf{Internals of \rdfframes.}
The internal workings of \rdfframes involve a logical representation of a query.
Query optimizers use some form of logical query representation, and we adopt a representation similar to the Query Graph Model~\cite{Pirahesh:1992}.
Another \rdfframes task is to generate \sparql queries from a logical representation.
This task is also performed by systems for federated \sparql query processing (e.g.,~\cite{fedx}) when they send a query to a remote site.
However, the focus in these systems is on answering \sparql triple patterns at different sites, so the queries that they generate are simple.
\rdfframes requires more complex queries so it cannot use federated \sparql  techniques.

\section{\rdfframes API}
\label{sec:data-model}

This section presents an overview of the \rdfframes API. \rdfframes provides the user with a set of operators, where each operator is implemented as a function in a programming language. 
Currently, this API is implemented in Python, but we describe the \rdfframes operators in generic terms since they
can be implemented 
in any programming language.
The goal of \rdfframes is to build a table (the dataframe) from a subset of information extracted from a knowledge graph. We start by describing the data model for a table constructed by \rdfframes, and then present an overview of 
the API operators.

\subsection{Data Model}
\label{subsec:data-model}

The main tabular data structure in \rdfframes
is called an \emph{\rdfframe}. 
This is the data structure constructed by API calls (\rdfframes operators). 
\rdfframes provides initialization operators that a user calls to initialize an \rdfframe and
other operators that extend or modify it.
Thus, an \rdfframe represents the data described by 
a sequence of one or more \rdfframes operators.
Since \rdfframes operators are not executed on relational tables but are mapped to \sparql graph patterns,
an \rdfframe is not represented as an actual
table in memory but rather as an abstract description of a table. 
A formal definition of a knowledge graph and an 
\rdfframe 
is as follows:

\begin{definition}[Knowledge Graph]
	A knowledge graph $G : (V, E)$ is a directed labeled RDF graph where the set of nodes $V \in I\cup L \cup B$ is a set of RDF URIs $I$, literals $L$, and blank nodes $B$ existing in $G$, and the set of labeled edges $E$ is a set of ordered pairs of elements of $V$ having labels from $I$.
	Two nodes connected by a labeled edge form a {\em triple} denoting the relationship between the two nodes. The knowledge graph is represented in \rdfframes by a 
	\textit{graph\_uri}.
\end{definition}

\begin{definition}[\rdfframe]
	Let $\mathbb{R}$ be the set of real numbers, $N$ be an infinite set of strings, and $V$ be the set of RDF URIs and literals.
	An \rdfframe $D$ is a pair $(\mathcal{C}, \mathcal{R})$, where $\mathcal{C} \subseteq N$ is an ordered set of column names of size $m$ and $\mathcal{R}$ is a bag of $m$-sized tuples with values from $V \cup \mathbb{R}$ denoting the rows. The size of
	$D$ is equal to the size of $\mathcal{R}$.
\end{definition}	

Intuitively, an \rdfframe is a subset of information extracted from one or more knowledge graphs. 
The rows of an \rdfframe should contain values that are either (a)~URIs or literals in a knowledge graph, or (b)~aggregated values on data extracted from a graph. 
Due to the bag semantics, an \rdfframe may contain duplicate rows,
which is good in machine learning because it preserves the data distribution and 
is compatible with the bag semantics of \sparql.

\subsection{API Operators}
\label{subsec:operators}

\rdfframes provides the user with two 
types of operators: (a)~exploration and navigational operators, which operate on a knowledge graph, and (b)~relational operators, which operate on an \rdfframe (or two in case of joins).
The full list of operators, and also other \rdfframes functions (e.g., for client-server communication), can be found with the source code.\footnote{\url{https://github.com/qcri/rdfframes}}

The \rdfframes exploration operators are needed to deal with one of the challenges of real-world knowledge graphs:
knowledge graphs in RDF are typically multi-topic, heterogeneous, incomplete, and sparse, and the data distributions can be highly skewed.
Identifying a relatively small, topic-focused dataset from such a knowledge graph
to extract into an \rdfframe is not a simple task, since it requires knowing the structure and schema of the dataset. \rdfframes provides data exploration operators to help with this task.
For example, \rdfframes includes operators to identify the RDF classes representing entity types
in a knowledge graph,
and to compute the data distributions of these classes.

Guided by the initial exploration of the graph, the user can gradually build an \rdfframe representing the
information to be extracted.
The first step
is always a call to 
the $seed$ operator (described below)
that initializes the \rdfframe with columns from the knowledge graph.
The rest of the \rdfframe is built through a sequence of calls to the \rdfframes navigational and relational operators.
Each of these operators outputs an \rdfframe.
The inputs to an operator can be a knowledge graph, one or more {\rdfframe}s, and/or other information such as predicates or column names.

The \rdfframes navigational operators are used to extract information from a knowledge graph into a tabular form using a navigational, procedural interface.
\rdfframes also provides relational operators that
apply operations on an \rdfframe such as filtering,
grouping, aggregation, filtering based on aggregate values, sorting, and join.
These operators do not access the knowledge graph, and one could argue that they are not necessary in \rdfframes since they are already provided by machine learning tools that work on dataframes such as pandas.
However, 
we opt to provide these operators
in \rdfframes 
so that they can be pushed into the RDF engine, which results in substantial performance gains as we will see in Section~\ref{sec:performance}.

In the following, we describe the syntax and semantics of the main operators of both types.
Without loss of generality, let $G=(V, E)$ be the input knowledge graph and $D=(\mathcal{C}, \mathcal{R})$ be the input \rdfframe of size $n$.
Let $D'=(\mathcal{C'}, \mathcal{R'})$ be the output \rdfframe.
In addition, let $\Join$, $\leftouterjoin$, $\rightouterjoin$, $\fullouterjoin$, $\sigma$, $\pi$, $\rho$, and $\gamma$ be the inner join, left outer join, right outer join, full outer join, selection, projection, renaming, and grouping-with-aggregation relational operators, respectively, defined using bag semantics as in typical relational databases~\cite{dbcomplete}.

\noindent{\textbf{Exploration and Navigational Operators.}}
These operators traverse a knowledge graph to extract information from it to either construct a new \rdfframe or expand an existing one. 
They bridge the gap between the RDF data model and the tabular format by allowing the user to extract tabular data through graph navigation.
They take as input either a knowledge graph $G$, or a knowledge graph $G$ and an \rdfframe $D$, and output an \rdfframe $D'$. 
\squishlist
\begin{sloppypar}
    \item $G.seed(col_1, col_2, col_3)$ where $col_1, col_2, col_3$ are in $N \cup V$: This operator is the starting point for constructing any \rdfframe. 
    Let $t = (col_1, col_2, col_3)$ be a \sparql triple pattern, then this operator creates an initial \rdfframe by converting the evaluation of the  triple pattern $t$ on graph $G$ to an \rdfframe. The returned \rdfframe has a column for every variable in the pattern $t$.
    Formally, let $D_t$ be the \rdfframe equivalent to the evaluation of the triple pattern $t$ on graph $G$. We formally define this notion of equivalence in Section~\ref{sec:correctness}.
	 The returned \rdfframe is defined as $D' = \pi_{N \cap \{col_1, col_2, col_3\}}(D_t)$.
	 As an example, the $seed$ operator can be used to retrieve all instances of class type $T$
	 in graph $G$ by calling $G.seed(instance, \text{\textit{rdf:type}}, T)$. 
	 For convenience, \rdfframes provides implementations for 
   the most common variants of this operator.
	 For example, the \texttt{feature\_domain\_range} operator in Listing~\ref{lst:motivating_code} initializes the \rdfframe with all pairs of entities in DBpedia connected by the predicate \texttt{dbpp:starring},
   which are movies and
	 the
	 actors starring in them.

  \item $(G, D).expand(col, pred, new\_col, dir, is\_opt)$, where $col \in \mathcal{C}$, $pred \in V$, $new\_col \in N$, $dir \in \{in, out\}$, and $is\_opt \in \{true, false\}$: 
	This is the main navigational operator in \rdfframes. It expands an \rdfframe by navigating from $col$ following the edge $pred$ to $new\_col$ in direction $dir$. 
	Depending on the direction of navigation, either the starting column for navigation $col$ is the subject of the triple and the ending column $new\_col$ is the object, or vice versa.
	$is\_opt$ determines whether null values are allowed. If $is\_opt$ is false, $expand$ filters out the rows in $D$ that have a null value in $new\_col$. Formally, if $t$ is a \sparql pattern representing the navigation step,
	then $t = (col, pred, new\_col)$ if direction is $out$ or $t = (new\_col, pred, col)$ if direction is $in$. 
	Let $D_{t}$ be the \rdfframe 
  corresponding to the evaluation of the triple pattern $t$ on graph G. $D_{t}$ will contain one column $new\_col$ and the rows are the objects of $t$ if the direction is $in$ or the subjects if the direction is $out$. 
  Then $D' = D \Join D_{t}$ if $is\_opt$ is false or $D' = D \leftouterjoin D_{t}$ if $is\_opt$ is true.
  For example, in Listing~\ref{lst:motivating_code}, \texttt{expand} is used twice,
  once to add the \texttt{country} attribute of the actor to the \rdfframe and once to find the movies and (if available) Academy Awards for prolific American actors.
\end{sloppypar}

\squishend

\noindent{\textbf{Relational Operators.}} These operators are used to clean and further process {\rdfframe}s.
They have the same semantics as in relational databases.
They take as input one or two {\rdfframe}s and output an \rdfframe.

\squishlist
	\item $D.filter(conds = [cond_1 \land cond_2 \land \ldots \land cond_k])$, where $conds$ is a list of expressions of the form $(col$ $\{<, >, =, \ldots\}$ $val)$ or one of the pre-defined boolean functions found in \sparql like $isURI(col)$ or $isLiteral(col)$: This operator filters out rows from an \rdfframe that do not conform to $conds$.
	Formally, let $\varphi = [cond_1 \land cond_2 \land \ldots \land cond_k$] be a propositional formula where $cond_i$ is an expression. 
	Then $D' = \sigma_{\varphi}(D)$.
	In Listing~\ref{lst:motivating_code}, \texttt{filter} is used two times, once to restrict the extracted data to American actors and once to restrict the results of a group by in order to identify prolific actors
  (defined as having 50 or more movies).
  The latter \texttt{filter} operator is applied
  after \texttt{group\_by} and the aggregation function \texttt{count}, which corresponds to a very different \sparql pattern compared to the first usage. However, this is handled internally by \rdfframes and is transparent to the user.

	\item $D.select\_cols(cols)$, where $cols \subseteq \mathcal{C}$: Similar to the relational projection operation, it keeps only the columns $cols$ and removes the rest. Formally, $D' = \pi_{cols}(D)$.

\begin{sloppypar}
	\item $D.join(D_2, col, col_2, jtype, new\_col)$, where $D_2 = (\mathcal{C}_2, \mathcal{R}_2)$ is another \rdfframe, $col \in \mathcal{C}$, $col_2 \in \mathcal{C}_2$, and $jtype \in \{\Join, \leftouterjoin, \rightouterjoin, \fullouterjoin\}$:
	This operator joins two \rdfframe tables on their columns $col$ and $col_2$ using the join type $jtype$. $new\_col$ is the desired name of the new joined column.
	Formally, $D' = \rho_{new\_col/col} (D)\ jtype\ \rho_{new\_col/col_2}(D_2)$.


	\item $D.group\_by(group\_cols).aggregation(fn, col, new\_col)$, where $group\_cols \subseteq \mathcal{C}$, $fn \in \{max, min, average, sum, count, sample\}$, $col \in \mathcal{C} $ and $new\_col \in N$: This operator groups the rows of $D$ according to their values in one or more columns $group\_cols$. As in the relational grouping and aggregation operation, it partitions the rows of an \rdfframe into 
  groups
  and then applies the aggregation function on the values of column $col$ within each group. It returns a new \rdfframe which contains the grouping columns and the result of the aggregation on each group, i.e., $\mathcal{C'} = group\_cols \cup \{new\_col\}$. The combinations of values of the grouping columns in $D'$ are unique. Formally,
	$D' = \gamma_{group\_cols, fn(col) \mapsto new\_col}(D)$.
	Note that query generation has special handling for {\rdfframe}s output by the $group\_by$ operator (termed \texttt{grouped {\rdfframe}s}).
	This special handling is internal to  \rdfframes and transparent to the user.
  In Listing~\ref{lst:motivating_code}, \texttt{group\_by} is used with the \texttt{count} function to find the number of movies in which each actor appears.
	\item $D.aggregate(fn, col, new\_col)$, where $col\in\mathcal{C}$ and $fn\in\{max, min, average, sum, count, distinct\_count\}$: This operator aggregates values of the column $col$ and returns an \rdfframe that has one column and one row containing the aggregated value. It has the same formal semantics
  as the $D.group\_by().aggregation()$ operator except that $group\_cols = \emptyset$, so the whole \rdfframe is assumed to be one group.
	No further processing can be done on the \rdfframe after this operator.
\end{sloppypar}

  \item $D.sort(cols\_order)$, where $cols\_order$  
  is a set of pairs $(col, order)$ with $col \in \mathcal{C}$ and $order \in \{asc, desc\}$: This operator sorts the rows of the \rdfframe according to their values in the given columns and their sorting order and returns
  a sorted \rdfframe.

	\item $D.head(k, i)$, where $k \leq n$: Returns the first $k$ rows of the \rdfframe starting from row $i$ (by default $i=0$).
	No further processing can be done on the \rdfframe after this operator.

\squishend

\vspace*{-8pt}
\section{Query Generation}
\label{sec:query-generation}

One of the key innovations in \rdfframes is the query generation process.
Query generation produces a \sparql query from an \rdfframe representing a sequence of calls to \rdfframes operators.
The guidelines we use in query generation to guarantee efficient processing are as follows:

\squishlist
\item
Include all of the computation required for generating an \rdfframe in the \sparql query sent to the RDF engine.
Pushing computation into the engine enables \rdfframes to take advantage of the benefits of a database system such as query optimization, bulk data processing, and near-data computing.

\item
Generate one \sparql query for each \rdfframe, never more. \rdfframes combines all graph patterns and operations described by an \rdfframe 
into a single \sparql query.
This minimizes the number of interactions with the RDF engine and enables the query optimizer to explore all optimization opportunities since it can see all operations.

\item
Ensure that the generated query is as simple as possible.
The query generation algorithm generates graph patterns that minimize the use of nested subqueries
and union 
\sparql patterns, since these are known to be expensive.
Note that, in principle, we are doing part of the job of the RDF engine's query optimizer.
A powerful-enough optimizer would be able to simplify and unnest queries whenever possible.
However, the reality is that \sparql is a complex language on which query optimizers do not always do a good job. As such, any steps to help the query optimizer are of great use. We show the performance benefit of this approach in Section~\ref{sec:performance}.

\item
Adopt a lazy execution model, generating and processing a query only when required by the user.

\item
Ensure that the generated \sparql query is correct, that is, ensure the query is semantically equivalent to the \rdfframe. We prove this in Section~\ref{sec:correctness}.

\squishend

\begin{figure}
  \centering
  \includegraphics[width=\columnwidth]{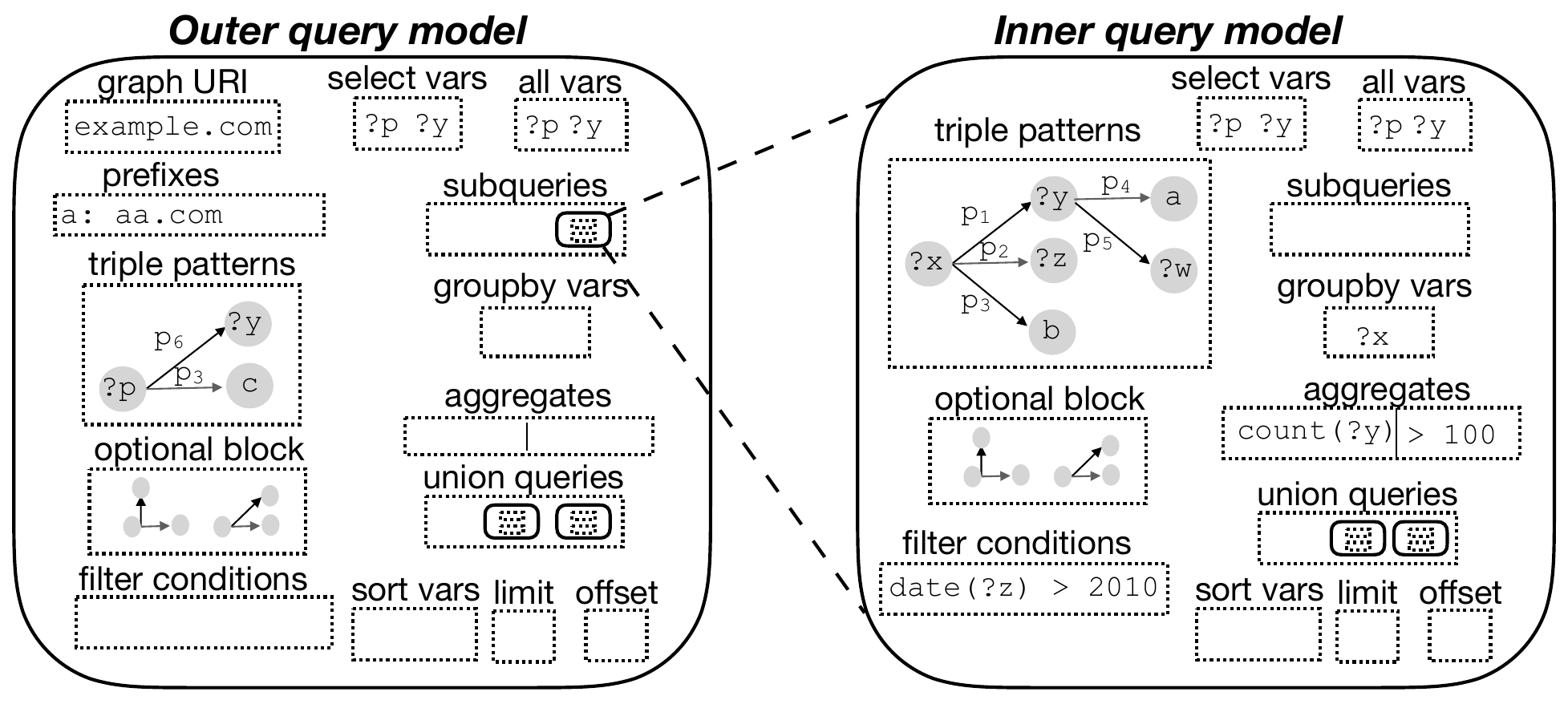}
  \caption{Example of an \rdfframes nested query model.}
  \label{fig:querymodel}
  \vspace*{-8pt}
\end{figure}


\vspace*{-3pt}
Our query model is inspired by the Query Graph Model~\cite{Pirahesh:1992}, and it
encapsulates all components required to construct a \sparql query. 
Query models can be nested in cases where nested subqueries are required.
Using the query model as an intermediate representation between an \rdfframe and the corresponding \sparql query allows for 
(i)~flexible implementation by separating the operator manipulation logic from the query generation logic, and (ii)~simpler optimization.
Without a query model, a naive implementation of \rdfframes would 
translate each operator to a \sparql pattern and encapsulate it in a subquery,
with one outer query joining all the subqueries to produce the result.
This is analogous to how some software-generated SQL queries are produced.
Other implementations are possible such as producing a \sparql query for
each operator and re-parsing it every time it has to be combined with a new pattern,
or directly manipulating the parse tree of the query.
The query model enables a simpler and more powerful implementation.

An example query model representing a nested \sparql query is shown in Figure~\ref{fig:querymodel}.
The left part of the figure is the outer query model, which has a reference to the inner query model (right part of the figure). 
The figure shows the components of a \sparql query represented in a query model.
These are as follows:
\squishlist

\item
Graph matching patterns including triple patterns, filter conditions, pointers to inner query models for sub-queries, optional blocks, and union patterns. Graph pattern matching is a basic operation in \sparql.
A \sparql query can be formed by combining triple patterns in various ways using different keywords.
The default is that a solution is produced if and only if
every triple pattern that appears in a graph pattern is matched to the triples in the RDF graph.
The OPTIONAL keyword adds triple patterns that extend the solution if they are matched, but do not eliminate the solution if they are not matched.
That is, OPTIONAL creates left outer join semantics.
The FILTER keyword adds a condition and restricts the query results to solutions that satisfy this condition.

\item
Aggregation constructs including: group-by columns, aggregation columns, and filters on aggregations (which result in a \texttt{HAVING} clause in the \sparql query). These patterns are applied to the result \rdfframe generated so far. Unlike graph matching patterns, they are not matched to the RDF graph.Aggregation constructs in inner query models are not propagated to outer query models.

\item
Query modifiers including limit, offset and sorting columns.
These constructs make final modifications to the result of the query. Any further API calls after adding these modifiers will result in a nested query as the current query model is wrapped and added to another query model. 

\item
The graph URIs by the query, the prefixes used, and the variables in the scope of each query.

\squishend

\vspace*{-9pt}
\subsection{Query Model Generation}

\vspace*{-2pt}

The query model is generated lazily, when the special \texttt{execute} function is called
on an \rdfframe.
We observe that generating the query model requires capturing the order of calls to \rdfframes operators and the parameters of these calls, but nothing more.
Thus, each \rdfframe $D$ created by the user is associated with a FIFO queue of operators.
The \texttt{Recorder} component of \rdfframes (recall Figure~\ref{fig:architecture})
records in this queue the sequence of operator calls made by the user.
When \texttt{execute} is called, the \texttt{Generator} component of \rdfframes creates the query model incrementally by processing the operators in this queue in FIFO order.
\rdfframes starts with an empty query model $m$. For each operator pulled from the queue of $D$, its corresponding \sparql component is inserted into $m$. 
Each \rdfframes operator edits one or two components of $m$.
\textit{All of the optimizations to generate efficient \sparql queries are done during query model generation.}

The first operator to be processed is always a $seed$ operator for which \rdfframes adds the corresponding triple pattern to the query model $m$. 
To process an $expand$ operator, it adds the corresponding triple pattern(s) to $m$. For example, the operator $expand(x, pred, y, \text{out}, \text{false})$ will result in the triple pattern $(?x, pred, ?y)$ being added to the triple patterns of $m$. Similarly, processing the $filter$ operator adds the conditions that are input parameters of this operator to the filter conditions in $m$.
To generate succinct optimized queries,  \rdfframes adds all triple and filter patterns to the same query model $m$, as long as the semantics are preserved. 
As a special case, when $filter$ is called on an aggregated column, the \texttt{Generator} adds the filtering condition to the $having$ component of $m$. 

One of the main challenges 
in designing \rdfframes was identifying the cases where a nested \sparql query is necessary. We were able to limit this to 
three cases where a nested query is needed to maintain the semantics:
\squishlist
    \item Case 1: when an $expand$ or $filter$ operator has to be applied on a grouped \rdfframe. The semantics here can be thought of as creating an \rdfframe that satisfies the expand or filter pattern and then joining it with the grouped \rdfframe. For example, the \rdfframes code in Listing~\ref{lst:case1_code} expands the \textit{country} column to obtain the \textit{continent} after the \texttt{group\_by} and \texttt{count}. This is semantically equivalent to building an \rdfframe of countries and their continents and then performing an inner join with the grouped \rdfframe.
    \vspace{8pt}
\begin{lstlisting}[
  aboveskip=-0.0\baselineskip,
  belowskip=-0.0\baselineskip,
  language=Python,
  breaklines=true,
  showspaces=false,
  basicstyle=\ttfamily\scriptsize,
  commentstyle=\color{gray},
  otherkeywords={expand, filter, group_by, join, feature_domain_range,entities,  .count, cache, unique},
  keywordstyle=\color{blue},
  caption={\rdfframes code - Expanding a grouped \rdfframe.},
  captionpos=b,
  label={lst:case1_code}]
df = graph.entities(':dpo:Actor', 'actor')\
  .expand('actor', [('dbp:birthPlace', 'country')])\
  .group_by(['actor'])\
  .count('country', 'country_count')\
  .expand('country', [('dbo:continent', 'continent'])
\end{lstlisting}

    \item Case 2: When a grouped \rdfframe has to be joined with another \rdfframe (grouped or non-grouped).
    For example, Listing~\ref{lst:case2_code} represents a join between a grouped \rdfframe and another \rdfframe.
\vspace{8pt}
\begin{lstlisting}[
  aboveskip=-0.0\baselineskip,
  belowskip=-0.0\baselineskip,
  language=Python,
  breaklines=true,
  showspaces=false,
  basicstyle=\ttfamily\scriptsize,
  commentstyle=\color{gray},
  otherkeywords={expand, filter, group_by, join, feature_domain_range,entities, count, cache, unique},
  keywordstyle=\color{blue},
  caption={\rdfframes code - Joining a grouped \rdfframe with another \rdfframe.},
  captionpos=b,
  label={lst:case2_code}]
df1 = graph.entities('dbo:Actor', 'actor')\
  .expand('actor', [('dbp:birthPlace', 'country')])\
  .group_by(['actor']).count('country', 'country_count')
df2 = graph.feature_domain_range('dbp:starring',
  'movie', 'actor').join(df1, 'actor', InnerJoin)
\end{lstlisting}
    \item Case 3: When two datasets are joined by a full outer join.
    For example, the \rdfframes code in Listing~\ref{lst:case3_code}
    is a full outer join between two datasests.
\vspace{8pt}
\begin{lstlisting}[
  aboveskip=-0.0\baselineskip,
  belowskip=-0.0\baselineskip,
  language=Python,
  breaklines=true,
  showspaces=false,
  basicstyle=\ttfamily\scriptsize,
  commentstyle=\color{gray},
  otherkeywords={expand, filter, group_by, join, feature_domain_range,entities, count, cache, unique},
  keywordstyle=\color{blue},
  captionpos=b,
  caption={\rdfframes code - Full outer join.},
  label={lst:case3_code}]
df1 = graph.entities('dpo:Actor', 'actor')\
  .expand('actor', [('dbp:birthPlace', 'country')])
df2 = graph.feature_domain_range('dbp:starring',
  'movie', 'actor').join(df1, 'actor', OuterJoin)
\end{lstlisting}

    There is no explicit full outer join between patterns in \sparql, only left outer join using the OPTIONAL pattern. Therefore, we define full outer join using the UNION and OPTIONAL patterns as the union of the left outer join and the right outer join of $D_1$ and $D_2$.
    A nesting query is required to wrap the query model for each \rdfframe inside the final query model.
\squishend

In the first case, when an $expand$ operation is called on a grouped \rdfframe, \rdfframes has to wrap the grouped \rdfframe in a nested subquery to ensure the evaluation of the grouping and aggregation operations before the expansion. 
\rdfframes uses the following steps to generate the subquery: (i)~create an empty query model $m'$, (ii)~transform the query model built so far $m$ into a subquery of $m'$, and~(iii)~add the new triple pattern from the $expand$ operator to the triple patterns of $m'$. In this case, $m'$ is the outer query model after the $expand$ operator and the grouped \rdfframe is represented by the inner query model $m$. Similarly, when $filter$ is applied on a grouping column in a grouped \rdfframe, \rdfframes creates a nested query model by transforming $m$ into a subquery. This is necessary since the filter operation was called after the aggregation and, thus, has to be done after the aggregation to maintain the correctness of the aggregated values.

The second case in which a nested subquery is required is when joining a grouped \rdfframe with another \rdfframe.
In the following, we describe in full the different cases of processing the $join$ operator, including the cases when subqueries are required.

To process the binary $join$ operator, \rdfframes needs to join two query models of two different {\rdfframe}s $D_1$ and $D_2$. If the join type is full outer join, a complex query that is equivalent to the full outer join
is constructed using the \sparql
OPTIONAL ($\leftouterjoin$) and UNION ($\cup$) patterns. 
Formally, $D_1 \fullouterjoin D_2 = (D_1 \leftouterjoin D_2)\ \cup \rho_{reorder} (D_2   \leftouterjoin D_1)$.

To process a full outer join, two new query models are constructed: 
The first query model ${m_1}{'}$ contains the left outer join of the query models $m_1$ and $m_2$, which represent $D_1$ and $D_2$, respectively.
The second query model ${m_2}{'}$ contains the right outer join of the of the query models $m_1$ and $m_2$, which is equivalent to the left outer join of $m_2$ and $m_1$. The columns of ${m_2}{'}$ are reordered to make them union compatible with ${m_1}{'}$. Nested queries are necessary to wrap the two query models $m_1$ and $m_2$ inside ${m_1}{'}$ and ${m_2}{'}$.
One final outer query model unions the two new query models ${m_1}{'}$ and ${m_2}{'}$.

For other join types, we distinguish three cases:
\squishlist
	\item $D_1$ and $D_2$ are not grouped: \rdfframes merges the two query models into one by combining their graph patterns (e.g., triple patterns and filter conditions). If the join type is left outer join, the patterns of $D_2$ are added inside a single OPTIONAL block of $D_1$. Conversely, for a right outer join the $D_1$ patterns are added as OPTIONAL in $D_2$. No nested query is generated here.
	\item $D_1$ is grouped and $D_2$ is not:
  \rdfframes merges the two query models via nesting. The query model of $D_1$ is the inner query model, while $D_2$ is set as the outer query model.  If the join type is left outer join, $D_2$ patterns are wrapped inside a single OPTIONAL block of $D_1$, and if the join type is right outer join, the subquery model generated for $D_1$ is wrapped in an OPTIONAL block in $D_2$. This is an example of the second case in which nested queries are necessary. The case when $D_2$ is grouped and $D_1$ is not is analogous to this case.
	\item Both $D_1$ and $D_2$ are grouped: \rdfframes creates one query model containing two nested query models, one for each \rdfframe, another example of the second case where nested queries are necessary. 
\squishend

If $D_1$ and $D_2$ are constructed from different graphs, the original graph URIs are used in the inner query to map each pattern to the graph it is supposed to match.

To process other operators such as $select\_cols$ and $group\_by$, \rdfframes fills the corresponding component in the query model. The $head$ operator maps to the $limit$ and
\textit{offset} components of the query model $m$. To finalize the join processing, \rdfframes unions the selection variables of the two query models, and takes the minimum of the offsets and the maximum of the limits (in case both query models have an offset and a limit).

\vspace*{-9pt}
\subsection{Translating to \sparql}
\begin{sloppypar}
The query model is designed to make translation to \sparql as direct and simple as possible. \rdfframes traverses a query model and translates each component of the model directly to the corresponding \sparql construct, following the syntax and style guidelines of \sparql.
For example, each prefix is translated to \texttt{PREFIX name\_space:name\_space\_uri}, graph URIs are added to the \texttt{FROM} clause, and each triple and filter pattern is added to the \texttt{WHERE} clause. The inner query models are translated recursively to \sparql queries and added to the outer query using the subquery syntax defined by \sparql.
When the query accesses more than one graph and different subsets of graph patterns are matched to different graphs, the \texttt{GRAPH} construct is used to wrap each subset of graph patterns with the matching graph URI. 
\end{sloppypar}

The generated \sparql query is sent to the RDF engine or \sparql endpoint using 
the \sparql protocol\footnote{\url{https://www.w3.org/TR/sparql11-protocol}} over HTTP.
We choose communication over HTTP since it is the most general mechanism to communicate with RDF engines and the only mechanism to communicate with \sparql endpoints.
One issue we need to address is paginating the results of a query, that is, retrieving them in chunks. There are several good reasons to paginate results, for example, avoiding timeouts at \sparql endpoints and bounding the amount of memory used for result buffering at the client. When using HTTP communication, we cannot rely on RDF engine cursors to do the pagination as they are engine-specific and not supported by the \sparql protocol over HTTP. The HTTP response returns only the first chunk of the result and the size of the chunk is limited by the \sparql endpoint configuration. The \sparql over HTTP client has to ask for the rest of the result chunk by chunk but this functionality is not implemented by many existing clients. 
Since our goal is generality and flexibility, \rdfframes implements pagination 
transparently to the user and returns one dataframe with all the query results.

\section{Semantic Correctness of Query Generation}
\label{sec:correctness}

In this section, we formally prove that the \sparql queries generated by \rdfframes return results that are consistent with the semantics of the \rdfframes operators.
We start with an overview of RDF and the \sparql algebra to establish the required notation.
We then summarize the semantics of \sparql, which is necessary for our correctness proof. Finally, we formally describe the query generation algorithm in \rdfframes and prove its correctness.

\subsection{\sparql Algebra}

The RDF data model can be defined as follows.
Assume there are countably infinite pairwise disjoint
sets $I$, $B$, and $L$ representing URIs, blank nodes, and literals, respectively. Let $T = (I \cup B \cup L)$ be the set of RDF terms. The basic component of an RDF graph is an RDF triple $(s,p,o) \in (I \cup B) \times I \times T $ where $s$ is the $subject$, $o$ is the $object$, and $p$ is the $predicate$.  An \textit{RDF graph} is a finite set of RDF triples.
Each triple represents a fact describing a relationship of type $predicate$ between the $subject$ and the $object$ nodes in the graph.

\sparql is a graph-matching query language that evaluates patterns on graphs and returns a result set. Its algebra consists of two
building blocks: \textit{expressions} and \textit{patterns}.

Let $X = \{?x_1, ?x_2, \ldots, ?x_n\}$ be a set of variables disjoint from the RDF terms $T$, the \sparql syntactic blocks are defined over $T$ and $X$. For a pattern $P$, $Var(P)$ is the set of all variables in $P$.
Expressions and patterns in are defined recursively as follows:
\squishlist
    \item A triple $ t \in (I \cup L \cup X) \times (I \cup X) \times (I \cup L \cup X)$ is a pattern.
    \item If $P_1$ and $P_2$ are patterns, then $ P_1\ Join\ P_2$, $ P_1\ Union\ P_2$, and  $P1\ LeftJoin\ P2$ are patterns. 
    \item Let all variables in $X$ and all terms in $I \cup L$ be \sparql expressions; then $ (E1+E2)$, $(E1-E2)$, $(E1\times E2)$, $ (E1/E2)$, $(E1=E2)$, $(E1<E2)$, $(\neg E1)$, $(E1\land E2)$, and $(E1\lor E2)$ are expressions. If $P$ is a pattern and $E$ is an expression then $Filter (E, P)$ is a pattern.
    \begin{sloppypar}
    \item If $P$ is a pattern and $X$ is a set of variables in $Var(P)$, then $Project (X, P)$ and $Distinct(Project(X, P))$ are patterns. These two constructs allow nested queries in \sparql and by adding them, there is no meaningful distinction between \sparql patterns and queries. 
    \end{sloppypar}
    \item If $P$ is a pattern, $E$ is an expression and $?x$ is a variable not in $Var(P)$, then $Extend(?x, E, P)$ is a pattern. This allows assignment of expression values to new variables and is used for variable renaming in \rdfframes.
    \begin{sloppypar}    
    \item If $X$ is a set of variables, $?z$ is another variable, $f$ is an aggregation function, $E$ is an expression, and $P$ is a pattern, then $GroupAgg(X,?z,f,E,P)$ is a pattern where $X$ is the set of grouping variables, $?z$ is a fresh variable to store the aggregation result, $E$ is often a variable that we are aggregating on. This pattern captures the grouping and aggregation constructs in \sparql 1.1. It induces a partitioning of a pattern's solution mappings into equivalence classes based on the values of the grouping variables and finds one aggregate value for each class using one of the aggregation functions in 
    $\{max, min, average, sum, count, sample\}$.
    \end{sloppypar}

\squishend
 
 \sparql defines some modifiers for the result set returned by the evaluation of the patterns. These modifiers include: $Order(X,order)$ where $X$ is the set of variables to sort on and $order$ is 
$ascending$ or $descending$, 
$Limit(n)$ which returns the first $n$ values of the result set, and
\textit{Offset}$(k)$
which returns the results starting from the $k$-th value.

\vspace*{-12pt}
\subsection{\sparql Semantics}
\label{subsec:sparql}
In this section, we present the semantics defined in~\cite{kaminski2016semantics}, which assumes bag semantics and integrates all \sparql1.1 features such as aggregation and subqueries.




The semantics of \sparql queries are based on multisets (bags) of mappings. A \textbf{mapping} is a partial function $\mu$ from $X$ to $T$ where $X$ is a set of variables and $T$ is the set of RDF terms. 
The domain of a mapping $dom(\mu)$ is the set of variables where $\mu$ is defined. $\mu_1$ and $\mu_2$ are compatible mappings, written ($\mu_1 \thicksim \mu_2$), if  $(\forall ?x \in dom(\mu_1) \cap dom(\mu_2), \mu_1(?x) = \mu_2(?x))$.
If $\mu_1 \thicksim \mu_2$, $\mu_1 \cup \mu_2$ is also a mapping and is obtained by extending $\mu_1$ by $\mu_2$ mappings on all the variables $dom(\mu_2) \setminus dom(\mu_1)$. 

A \sparql pattern solution is a multiset $\Omega = (S_{\Omega}, card_{\Omega})$ where $S_{\Omega}$ is the base set of mappings, and the multiplicity function $card_{\Omega}$ assigns a positive number to each element of $S_{\Omega}$.

Let $\llbracket E \rrbracket_G$ denote the evaluation of expression $E$ on graph $G$, $\mu(P)$ the pattern obtained from $P$ by replacing its variables according to $\mu$,
and $Var(P)$ the set of all the variables
in $P$.
The semantics of patterns over graph $G$ are defined as:
\squishlist
    \item  $\llbracket t \rrbracket_G$: the solution of a triple pattern $t$ is the multiset with $S_t = $ all $\mu$ such that $dom(\mu) = Var(t)$ and $\mu(t) \in G$. $card_{\llbracket t \rrbracket_G}(\mu) = 1$ for all such $\mu$.
    \item $\llbracket P_1\ Join\ P_2 \rrbracket_G = \{\{
    \mu | \mu_1 \in \llbracket P1 \rrbracket_G, \mu_2 \in \llbracket P2 \rrbracket_G, \mu = \mu_1 \cup \mu_2
    \}\}$
    \item $\llbracket P_1\ LeftJoin\ P_2 \rrbracket_G = 
    \{\{
    \mu | \mu \in \llbracket P_1\ Join\ P_2 \rrbracket_G
    \}\} \uplus $\\
    $\{\{\mu | \mu \in \llbracket P1 \rrbracket_G, \forall \mu_2 \in \llbracket P2 \rrbracket_G, (\mu \nsim \mu_2)
    \}\}$
    \item $\llbracket P1\ Union\ P2 \rrbracket_G =  \llbracket P1 \rrbracket_G \uplus	 \llbracket P2 \rrbracket_G$
    \item $\llbracket Filter(E, P)\rrbracket_G = \{\{    \mu | \mu \in \llbracket P1 \rrbracket_G, \llbracket E \rrbracket_{\mu,G} =true \}\}$
    \item $\llbracket Project(X, P)\rrbracket_G = $   $\forall \mu \in \llbracket P\rrbracket_G$, if $\mu$ is a restriction to $X$ then it is in the base set of this pattern and its multiplicity is
    the
    sum of multiplicities of all corresponding $\mu$.
    \begin{sloppypar}
    \item $\llbracket Distinct(Q) \rrbracket_G = $  the multiset with the same base set as $\llbracket Q \rrbracket_G$, but with multiplicity 1 for all mappings.
The \sparql patterns $Project (X, P)$ and $Distinct(Project(X, P))$ define a \sparql query.
When used in the middle of a query, they define a nested query.
    \end{sloppypar}
    \begin{sloppypar}
    \item $\llbracket Extend(?x, E, P)\rrbracket_G =$\\ $\{\mu'| \mu \in \llbracket P\rrbracket_G, \mu'=\mu\cup\{?x\rightarrow \llbracket E\rrbracket_{\mu,G}\}, \llbracket E\rrbracket_{\mu,G} \neq Error \} \uplus \\ \{\mu| \mu \in \llbracket P\rrbracket_G, \llbracket E \rrbracket_{\mu,G} =  Error \} $ and $Var(Extend(?x, E, P)) = \{?x\} \cap Var(P)$
    \end{sloppypar}
    \item Given a graph $G$, let $v|x $ be the restriction of $v$ to $X$, then $\llbracket GroupAgg(X,?z,f,E,P)\rrbracket_G$ is the multiset with the base set:\\
    $\{\mu' | \mu' = \mu|X \cup \{?z \rightarrow v_{\mu}\}, \mu \in \llbracket P\rrbracket_G, v_{\mu} \neq Error  \} \cup \newline \{\mu' | \mu' = \mu|X , \mu \in \llbracket P \rrbracket_G, v_{\mu} = Error  \}$ and multiplicity 1 for each mapping in the base set, where for each mapping $\mu \in \llbracket P \rrbracket_G$, the value of the aggregation function on the group that the mapping belongs to is $v_{\mu} = f(\{v\ |\ \mu' \in \llbracket P \rrbracket_G, \mu'|x = \mu|x, v = \llbracket E \rrbracket_{\mu',G} \})$.
\squishend

\begin{table*}[th!]
    \caption{Mappings of \rdfframes operators on graph $G$ and/or
    \rdfframe $D$ to \sparql patterns on $G$. $P$ is the \sparql pattern equivalent to the sequence of \rdfframes operators called so far on $D$ (or null if $D$ is new).  \label{tab:mappings}}
    \begin{tabular}{ |c|c| } 
        \hline
        \textbf{\rdfframes Operator $\bm{O}$} & \textbf{SPARQL pattern: $\bm{g(O, P)}$} \\
        \hline
        $seed(col_1, col_2, col_3)$ & $Project(Var(t),t),$ where  $t = (col_1, col_2, col_3)$\\
        \hline
        $expand(x, pred, y, out, false)$ & $P \Join (?x, pred, ?y) $ \\
        $expand(x, pred, y, in, false)$ & $P \Join (?y, pred, ?x) $ \\ 
        $expand(x, pred, y, out, True)$ & $P \leftouterjoin (?x, pred, ?y) $\\
        $expand(x, pred, y, in, True)$ & $P \leftouterjoin (?y, pred, ?x) $\\
        \hline
        $join(D_2, col, col_2, \Join, new\_col)$ & $Extend(new\_col, col, P) \Join Extend(new\_col, col_2, P_2)$, $P_2 = F(O_{D_2})$\\
        $join(D_2, col, col_2, \leftouterjoin, new\_col)$ & $Extend(new\_col, col, P) \leftouterjoin Extend(new\_col, col_2, P_2) $, $P_2 = F(O_{D_2})$\\
        $join(D_2, col, col_2, \rightouterjoin, new\_col)$ & $Extend(new\_col, col_2, P_2) \leftouterjoin Extend(new\_col, col, P) $, $P_2 = F(O_{D_2})$\\
        $join(D_2, col, col_2, \fullouterjoin, new\_col)$ & $(P1 \leftouterjoin P_2) \cup (P_2 \leftouterjoin P_1)$, \\
        & $P_1 = Extend(new\_col, col, P)$, $P_2 = Extend(new\_col, col_2, F(O_{D_2})) $\\
               \hline
        $filter(conds = [cond_1 \land cond_2 \land \dots \land cond_k])$ & $Filter(conds, P)$\\
        \hline
        $select\_cols(cols)$ & $Project(cols, P)$\\
        \hline
        $groupby(group\_cols).$ & $Project(group\_cols \cup \{new\_col\}, $\\ 
        $aggregation(fn, src\_col, new\_col)$& $GroupAgg(group\_cols, new\_col, fn, src\_col, P))$ \\
        \hline
        $aggregate(fn, col, new\_col)$ & $Project(\{ new\_col \}, GroupAgg(\emptyset, new\_col, fn, col, P))$ \\
        \hline
    \end{tabular}
\end{table*}

\vspace*{-4pt}
\subsection{Semantic Correctness}
\label{subsec:correctness}

Having defined the semantics of \sparql patterns, we now prove the semantic correctness of query generation in \rdfframes as follows.
First, we formally define the \sparql query generation algorithm. That is, we define the \sparql query or pattern generated by any sequence of \rdfframes operators. We then  prove that the solution sets of the generated \sparql patterns are equivalent to the \rdfframes tables defined by the semantics of the sequence of \rdfframes operators.


\subsubsection{Query Generation Algorithm}
To formally define the query generation algorithm, we first define the \sparql pattern each \rdfframes operator generates. 
We then give a recursive definition of a non-empty RDFFrame and then define a recursive mapping from any sequence of \rdfframes operators constructed by the user to a \sparql pattern using the patterns generated by each operator. This mapping is based on the query model 
described in Section~\ref{sec:query-generation}.

\vspace*{-8pt}
\begin{definition}[Non-empty RDFFrame]
\label{def:recursivedataset}
	A non-empty RDFFrame is either generated by the $seed$ operator or by applying an \rdfframes operator on one or two non-empty RDFFrames. 
\end{definition}

\vspace*{-4pt}
Given a non-empty RDFFrame $D$, let $O_D$ be the sequence of \rdfframes operators that generated it. 

\vspace*{-4pt}
\begin{sloppypar}
\begin{definition}[Operators to Patterns]
	Let $O = [o_1, \dots, o_k]$ be a sequence of \rdfframes operators and $P$ be a \sparql pattern.
	Also let $g: (o, P) \rightarrow P$ be the mapping from a single \rdfframes operator $o$ to a \sparql pattern based on the query generation of \rdfframes described in Section~\ref{sec:query-generation}, also illustrated in Table~\ref{tab:mappings}.
    Mapping $g$ takes as input an \rdfframes operator $o$ and a \sparql pattern $P$ corresponding to the operators done so far on an \rdfframe $D$, applies a \sparql operator defined by the query model generation algorithm on the input \sparql pattern $P$, and returns a new \sparql pattern.
	Using $g$, we define a recursive mapping $F$ on a sequence of \rdfframes operators $O$, $F: O \rightarrow P$, as: 
    \begin{equation}
      F(O)=\begin{cases}
        g(o_1, Null), & \text{if $|O| \leq 1$}.\\
        g(o_k, F(O_{[1:k-1]}), F(O_{D_2})), & {o_k = join(D_2,\ldots)}.\\
        g(o_k, F(O_{[1:k-1]})), & \text{otherwise}.
      \end{cases}
    \end{equation}
\end{definition}
\end{sloppypar}
$F$ returns a triple pattern for the seed operator and then builds the rest of the \sparql query by iterating over the \rdfframes operators according to their order in the sequence $O$.

\subsection{Proof of Correctness}
\vspace*{-3pt}
To prove the equivalence between the \sparql pattern solution returned by $F$ and the \rdfframe generating it, 
we first define the meaning of equivalence between a relational table with bag semantics and the solution sets 
of
\sparql queries.
First, we define a mapping that converts \sparql solution sets to relational tables by letting the domains of the mappings be the columns and their ranges be the rows.
Next, we define the equivalence between solution sets and relations.

\vspace*{-4pt}
\begin{sloppypar}
\begin{definition}[Solution Sets to Relations]
\label{def:solution2relation}
Let $\Omega = (S_{\Omega}, card_{\Omega})$ be a multiset (bag) of mappings returned by  the evaluation of a \sparql pattern and $Var(\Omega) = \{?x; ?x \in dom(\mu), \forall \mu \in S_{\Omega} \}$ be the set of variables in $\Omega$.
Let ${L} = order(Var(\Omega))$ be the ordered set of elements in $Var(\Omega)$.
We define a conversion function $\lambda$: $\Omega$ $\rightarrow R$, where $R = ({C}, {T})$ is a relation. R is defined such that its ordered set of columns (attributes) are the variables in $\Omega$ (i.e., ${C}=L$), and ${T} = (S_{{T}}, card_{T})$ is a multiset of (tuples) of values such that for every $\mu$ in $S_{\Omega}$, there is a tuple $\tau \in S_{{T}}$ of length $n = |(Var(\Omega))|$ and $\tau_i = \mu(L_i)$. The multiplicity function $(card_{T})$ is defined such that the multiplicity of $\tau$ is equal to the multiplicity of $\mu$ in $card_{\Omega}$.
\end{definition}
\end{sloppypar}

\vspace*{-4pt}
\begin{definition}[Equivalence]
\label{def:equivalence}
A \sparql pattern solution $\Omega = (S_{\Omega}, card_{\Omega})$ is {\em equivalent} to a relation $R = ({C}, {T})$, written($\Omega \equiv R$), if and only if $R=\lambda(\Omega)$.
\end{definition}

\vspace*{-4pt}
We are now ready to use this definition to present a lemma that defines the equivalent relational tables for the  main \sparql patterns used in our proof.

\vspace*{-4pt}
\begin{lemma}
\label{lem:equivofoperations}
If $P_1$ and $P_2$ are SPARQL patterns, then:
\begin{packed_enum}
\item[a.]
$\llbracket(P_1\ Join\ P_2)\rrbracket_G \equiv \lambda(\llbracket P_1\rrbracket_G) \Join \lambda(\llbracket P_1\rrbracket_G)$,
\item[b.]$\llbracket(P_1\ LeftJoin \ P_2)\rrbracket_G \equiv \lambda(\llbracket P_1\rrbracket_G) \leftouterjoin \lambda(\llbracket P_1\rrbracket_G)$,
\item[c.]$\llbracket(P_1\ Union \ P_2)\rrbracket_G \equiv \lambda(\llbracket P_1\rrbracket_G) \fullouterjoin \lambda(\llbracket P_1\rrbracket_G)$
\item[d.]$\llbracket(Extend(?x, E, P))\rrbracket_G \equiv \rho_{?x/E}(\lambda(\llbracket P\rrbracket_G)$
\item[e.]
$\llbracket(Filter(conds, P))\rrbracket_G \equiv 
\sigma_{conds}(\lambda(\llbracket P\rrbracket_G))$
\item[f.]$\llbracket(Project(cols, P))\rrbracket_G \equiv 
\Pi_{cols}(\lambda(\llbracket P\rrbracket_G))$
\item[g.]$\llbracket(GroupAgg(\emptyset, new\_col, fn, col, P)))\rrbracket_G \equiv$\\ $\gamma_{cols, fn(col) \mapsto new\_col}(\lambda(\llbracket P\rrbracket_G))$
\end{packed_enum}
\end{lemma}

\begin{proof}
The proof of this lemma follows from
(1)~the semantics of \sparql operators presented in Section~\ref{subsec:sparql},
(2)~the well-known semantics of relational operators,
(3)~Definition~\ref{def:solution2relation} which specifies the function $\lambda$,
and
(4)~Definition~\ref{def:equivalence} which defines the equivalence between multisets of mappings and relations.
For each statement in the lemma, we use the definition of the function $\lambda$, the relational operator semantics, and the \sparql operator semantics to define the relation on the right side.
Then we use the definition of \sparql operators semantic to define the 
multiset on the left side. 
Finally,
Definition~\ref{def:equivalence} proves the statement.
\end{proof}

\vspace*{-4pt}
Finally, we present the main theorem in this section, which guarantees the semantic correctness of the \rdfframes query generation algorithm.

\vspace*{-4pt}
\begin{sloppypar}
\begin{theorem}
	Given a graph $G$, every RDFFrame $D$ that is returned by a sequence of \rdfframes operators $O_D = [o_1, \dots, o_k]$ on $G$ is equivalent to the evaluation of the \sparql pattern $P = F(O_D)$ on G. 
	In other words, $D\equiv\llbracket F(O_D)\rrbracket_G$.
\end{theorem}
\end{sloppypar}

\begin{proof}
We prove that $D\equiv\lambda(\llbracket F(O_D)\rrbracket_G)$ via structural induction on non-empty RDFFrame $D$.
For simplicity, we denote the proposition $D\equiv\llbracket F(O_D)\rrbracket_G$ as $A(D)$.

\noindent {\em Base case:} Let $D$ be an RDFFrame created by one \rdfframes operator $O_D = [seed(col_1, col_2, col_3)]$. The first operator (and the only one in this case) has to be the $seed$ operator since it is the only operator that takes only a knowledge graph as input and returns an RDFFrame. From Table~\ref{tab:mappings}:
\begin{align*}
    F(O_D) &= g(seed(col_1, col_2, col_3), Null) \\
          &=Project(Var(t),t)
\end{align*}
\noindent where $t = (col_1, col_2, col_3)$. By definition of the \rdfframes operators in Section~\ref{sec:data-model}, $D = \Pi_{X \cap \{col_1, col_2, col_3\}}$ $(\lambda(\llbracket (t) \rrbracket_G))$ and by Lemma~\ref{lem:equivofoperations}(f), $A(D)$ holds.

\noindent {\em Induction hypothesis:} Every \rdfframes operator takes as input one or two RDFFrames $D_1, D_2$ and outputs an RDFFrame $D$. 
Without loss of generality, assume that both $D_1$ and $D_2$ are non-empty and $A(D_1)$ and $A(D_2)$ hold, i.e., $D_1\equiv\llbracket F(O_{D_1})\rrbracket_G$ and 
$D_2\equiv\llbracket F(O_{D_2})\rrbracket_G$.

\noindent {\em Induction step:} Let $D = D_1.Op(\text{optional } D_2)$, $P_1 = F(O_{D_1})$, and $P_2 = F(O_{D_2})$.
We use \rdfframes semantics to define $D$, the mapping $F$ to define the new pattern $P$, then 
Lemma~\ref{lem:equivofoperations} to prove the equivalence between $F$ and $D$.
We present the different cases next.
\squishlist
    \item If $Op$ is $expand(x, pred, y, out, false)$ then: $D = D_1 \Join \lambda(\llbracket t\rrbracket_G)$ according to the definition of the operator in Section~\ref{subsec:operators} and Table~\ref{tab:mappings}, where $t$ is the triple pattern $(?x$, $pred$, $?y)$. By the induction hypothesis, it holds that $D_1 = \lambda(\llbracket P_1)\rrbracket_G)$. 
    Thus, it holds that $D = \lambda(\llbracket P_1\rrbracket_G) \Join \lambda(\llbracket t \rrbracket_G)$ and by Lemma~\ref{lem:equivofoperations}(a), $A(D)$ holds.
    The same holds when $Op$ is $expand(x, pred, y, in, false)$ except that $t=(?y, pred, ?x)$.


    \begin{sloppypar}
    \item If $Op$ is $join(D_2, col, col_2, \Join, new\_col)$ then:
    $D = \rho_{new\_col/col}(D_1) \Join \rho_{new\_col/col_2}(D_2) $, and by $A(D_1)$, $D_1 = \lambda(\llbracket P_1\rrbracket_G)$ and $D_2 = \lambda(\llbracket P_2\rrbracket_G)$. Thus, $D = \rho_{new\_col/col}\lambda(\llbracket P_1\rrbracket_G)) 
    \Join \rho_{new\_col/col_2}(\lambda(\llbracket P_2\rrbracket_G))$ and by  Lemma~\ref{lem:equivofoperations}(a,c), $A(D)$ holds. The same argument holds for other types of join, using the relevant parts of Lemma~\ref{lem:equivofoperations}.
    \end{sloppypar}
    
    \item If $Op$ is $filter(conds = [cond_1 \land cond_2 \land \dots \land cond_k])$ then: $D = \sigma_{conds}(D_1)$, and by $A(D_1)$, $D_1 = \lambda(\llbracket P_1\rrbracket_G)$. So, $D = \sigma_{conds}\lambda(\llbracket P_1\rrbracket_G))$ and by Lemma~\ref{lem:equivofoperations}(e), $A(D)$ holds.

    
    \item If $Op$ is $groupby(cols).aggregation(f, col, new\_col)$ then: $D = \gamma_{cols, f(col) \mapsto new\_col}(D_1)$, and by $A(D_1)$, $D_1 = \lambda(\llbracket P_1\rrbracket_G)$. So, $D = \gamma_{cols, f(col) \mapsto new\_col}\lambda(\llbracket P_1\rrbracket_G))$ and by  Lemma~\ref{lem:equivofoperations}(f,g), $A(D)$ holds.
    
\squishend
Thus, $A(D)$ holds in all cases.
\end{proof}

\section{Evaluation}
\label{sec:performance}

We present an experimental evaluation of \rdfframes
in which our goal is to answer two questions:
(1)~How effective are the design decisions made in \rdfframes? 
and
(2)~How does \rdfframes perform compared to alternative baselines?

We use two workloads for this experimental study.
The first is made up of three case studies consisting of   
machine learning tasks on two real-world knowledge graphs. 
Each task starts with a data preparation step that extracts a pandas dataframe from the knowledge graph.
This step is the focus of the case studies.
In the next section, we present the \rdfframes code for each case study and the corresponding \sparql query generated by \rdfframes.
As in our motivating example, we will see that the \sparql queries are longer and more complex than the \rdfframes code, thereby showing that \rdfframes can indeed simplify access to knowledge graphs.
The full Python code for the case studies can be found in Appendix~\ref{app:fullcode}.
The second workload in our experiments is a synthetic workload consisting of 16 queries.
These queries are designed to exercise different features of \rdfframes for the purpose of benchmarking.
We describe the two workloads next, followed by the experimental setup and the results.
\vspace*{-7pt}
\subsection{Case Studies}
\label{subsec:case-studies}
\vspace*{8pt}
\subsubsection{Movie Genre Classification}
\label{subsec:genreclass}

\vspace*{-10pt}
Classification is a basic supervised machine learning task. 
This case study applies a classification task on movie data extracted from the DBpedia
knowledge graph.
Many knowledge graphs, including DBpedia, are heterogeneous,
with information about diverse topics, so extracting a
topic-focused dataframe for classification is challenging.

This task uses \rdfframes to build a dataframe of movies from DBpedia, along with
a set of movie attributes that can be used for 
movie genre classification.
The task bears some similarity to the code in Listing~\ref{lst:motivating_code}.
Let us say that the classification dataset that we want
includes movies that star
American actors (since they are assumed to have a global reach)
or prolific actors (defined as those who have starred in 100 or more movies).
We want the movies starring these actors, and for each movie, we extract the movie name (i.e., title), actor name, topic, country of production, and genre.
Genre is not always available so it is an optional predicate.
The full code for this data preparation step is shown in Listing~\ref{lst:genre}, and the \sparql query generated by \rdfframes is shown in Listing~\ref{lst:genresparql}.



The extracted dataframe can be used as a classification dataset by any popular Python machine learning library.
The movies that have the genre available in the dataframe can be used as labeled training data to train a classifier.
The features for this classifier are the attributes of the movies and the actors, and the classifier is trained to predict the genre of a movie based on these features.
The classifier can then be used to predict the genres of all movies that are missing the genre.

Note that the focus of \rdfframes is the data preparation step of a machine learning pipeline
(i.e.,~creating the dataframe).
That is, \rdfframes addresses the following problem: Most machine learning pipelines require as their starting point an input dataframe,
and there is no easy way to get such a dataframe from a knowledge graph while leveraging an RDF engine.
Thus, the focus of \rdfframes is enabling the user to obtain a dataframe from an RDF engine, and not how the machine learning pipeline uses this dataframe.
Nevertheless, it is interesting to see this dataframe within an end-to-end  machine learning pipeline.
Specifically, for the current case study,
can the dataframe created by \rdfframes be used for movie genre classification?
We emphasize that the accuracy of the classifier is not our main concern here;
our concern is demonstrating \rdfframes in an end-to-end machine learning pipeline.
Issues such as using a complex classifier, studying feature importance, or analyzing the distribution of the retrieved data are beyond the scope of \rdfframes.

To show \rdfframes in an end-to-end machine learning pipeline,
we built a classifier based on the output of
Listing~\ref{lst:genre}
to classify the six most frequent movie genres,
specifically, drama, sitcom, science fiction, legal drama, comedy, and fantasy.
The classification dataset consisted of 7,635 movies that represent the English movies in these six movie genres.
We trained a random forest classifier using the scikit-learn machine learning library based on movie features such as actor country, movie country, subject, and actor name.
This classifier achieved 92.4\% accuracy on evaluation data that is 30\% of the classification dataset.

We performed a similar experiment on song data from DBpedia. 
We extracted 27,956 triples of English songs in DBpedia along with their features such as album, writer, title, artist, producer, album title, and studio. 
We used the same methodology as in the movie genre classification case study to classify songs into genres such as alternative rock, hip hop, indie rock, and pop-punk.
The accuracy achieved by a random forest classifier in this case was 70.9\%.

\vspace{8pt}
\begin{lstlisting}[
aboveskip=-0.0\baselineskip,
belowskip=-0.0\baselineskip,
language=Python,
showspaces=false,
basicstyle=\ttfamily\scriptsize,
commentstyle=\color{gray},
otherkeywords={expand, filter, group_by, join, feature_domain_range, .count, cache, unique},
keywordstyle=\color{blue},
caption={\rdfframes code -  Movie genre classification.},
captionpos=b,
label={lst:genre}]
movies = graph.feature_domain_range('dbpp:starring', 'movie', 'actor')
movies = movies.expand('actor',[('dbpp:birthPlace',
    'actor_country'), ('rdfs:label', 'actor_name')])\
        .expand('movie', [('rdfs:label', 'movie_name'),
        ('dcterms:subject', 'subject'),
        ('dbpp:country', 'movie_country'),
        ('dbpo:genre', 'genre', Optional)]).cache()
american = movies.filter({'actor_country':\
 ['=dbpr:UnitedStates']})
prolific = movies.group_by(['actor'])\
        .count('movie', 'movie_count', unique=True)\
        .filter({'movie_count': ['>=100']})
dataset = american.join(prolific,'actor', OuterJoin)\
        .join(movies, 'actor', InnerJoin)
\end{lstlisting}

\begin{lstlisting}[
  aboveskip=-0.0\baselineskip,
  belowskip=-0.0\baselineskip,
  language=SQL,
  breaklines=true,
  showspaces=false,
  basicstyle=\ttfamily\scriptsize,
  commentstyle=\color{gray},
  keywordstyle=\color{blue},
  otherkeywords={OPTIONAL, FILTER},
  caption={\sparql query generated by \rdfframes for the code shown in Listing~\ref{lst:genre}.},
  captionpos=b,
  label={lst:genresparql}
  ]
  SELECT DISTINCT  ?actor_name ?movie_name ?actor_country ?genre ?subject
  FROM <http://dbpedia.org>
  WHERE
  { ?movie  dbpp:starring  ?actor .
    ?movie  rdfs:label  ?movie_name .
    ?movie dcterms:subject  ?subject .
    ?actor dbpp:birthPlace  ?actor_country .
    ?actor rdfs:label ?actor_name
    OPTIONAL
    { ?movie  dbpp:genre  ?genre }
    {{ SELECT  * WHERE
      {{ SELECT  * WHERE
        { ?movie  dbpp:starring  ?actor .
          ?movie rdfs:label  ?movie_name .
          ?movie dcterms:subject  ?subject .
          ?actor dbpp:birthPlace  ?actor_country .
          ?actor  rdfs:label  ?actor_name
          FILTER regex(str(?actor_country), "USA")
          OPTIONAL 
            { ?movie  dbpp:genre ?genre }
          }
      }
  OPTIONAL
    { SELECT DISTINCT ?actor (COUNT(DISTINCT ?movie) AS ?movie_count)
      WHERE
          { ?movie  dbpp:starring ?actor .
          ?movie  rdfs:label  ?movie_name .
          ?movie  dcterms:subject ?subject .
          ?actor  dbpp:birthPlace ?actor_country .
          ?actor  rdfs:label ?actor_name
          OPTIONAL 
             { ?movie  dbpp:genre ?genre }
          }
          GROUP BY ?actor
          HAVING ( COUNT(DISTINCT ?movie) >= 100 )
            }
          }
        }
  UNION
  { SELECT  * WHERE
   {{ SELECT DISTINCT ?actor (COUNT(DISTINCT ?movie) AS ?movie_count) WHERE
     { ?movie  dbpp:starring ?actor .
       ?movie  rdfs:label  ?movie_name .
       ?movie  dcterms:subject ?subject .
       ?actor  dbpp:birthPlace ?actor_country .
       ?actor  rdfs:label ?actor_name
      OPTIONAL 
        { ?movie  dbpp:genre  ?genre }
      }
      GROUP BY ?actor
      HAVING ( COUNT(DISTINCT ?movie) >= 100 )
      }
      OPTIONAL
      { SELECT  * WHERE
        { ?movie  dbpp:starring ?actor .
          ?movie  rdfs:label  ?movie_name .
          ?movie  dcterms:subject ?subject .
          ?actor  dbpp:birthPlace ?actor_country .
          ?actor  rdfs:label ?actor_name
          FILTER regex(str(?actor_country), "USA")
          OPTIONAL
            { ?movie  dbpp:genre ?genre }
          }
        }
      }
    }
  }
}
  \end{lstlisting}

  \vspace*{-6pt}
\subsubsection{Topic Modeling}
\label{subsec:topicmodeling}

\vspace*{-2pt}
Topic modeling is a statistical technique commonly used to identify hidden contextual topics in the text. 
In this case study, we use topic modeling to identify the active topics of research
in the database community.
We define these as the topics of recent papers published by authors who have published many SIGMOD and VLDB papers.
This is clearly an artificial definition, but it enables us to study the capabilities and performance of \rdfframes.
As stated earlier,  we are focused on data preparation not 
the details of the machine learning task.

\begin{lstlisting}[
aboveskip=0.5\baselineskip,
belowskip=-0.0\baselineskip,
language=Python,
showspaces=false,
breaklines=true,
basicstyle=\ttfamily\scriptsize,
commentstyle=\color{gray},
otherkeywords={expand, filter, group_by, join, entities, .count, cache, unique, select_cols},
keywordstyle=\color{blue},
caption={\rdfframes code - Topic modeling.},
captionpos=b,
label={lst:topic}]
papers = graph.entities('swrc:InProceedings','paper')
papers = papers.expand('paper',[('dc:creator',\
    'author'), ('dcterm:issued', 'date'),\
    ('swrc:series', 'conference'),\
        ('dc:title', 'title')]).cache()
authors = papers.filter({'date': ['>=2000'],
'conference': ['In(dblp:vldb, dblp:sigmod)']})
        .group_by(['author']).count('paper', 'n_papers')\
        .filter({'n_papers': '>=20', 'date': ['>=2010']})\
titles = papers.join(authors, 'author', InnerJoin)\
        .select_cols(['title'])
\end{lstlisting}

\begin{lstlisting}[
  aboveskip=-0.0\baselineskip,
  belowskip=-0.0\baselineskip,
  language=SQL,
  showspaces=false,
  breaklines=true,
  basicstyle=\ttfamily\scriptsize,
  commentstyle=\color{gray},
  keywordstyle=\color{blue},
  otherkeywords={OPTIONAL, FILTER},
  caption={\sparql query generated by \rdfframes for the code shown in Listing~\ref{lst:topic}.},
  captionpos=b,
  label={lst:topicsparql}
  ]
  SELECT  ?title
  FROM <http://dblp.l3s.de>
  WHERE
    { ?paper  dc:title       ?title ;
              rdf:type       swrc:InProceedings ;
              dcterm:issued  ?date ;
              dc:creator     ?author 
      FILTER ( year(xsd:dateTime(?date)) >= 2005 )
      { SELECT  ?author
        WHERE
          { ?paper  rdf:type       swrc:InProceedings ;
                    swrc:series    ?conference ;
                    dc:creator     ?author ;
                    dcterm:issued  ?date
            FILTER ( ( year(xsd:dateTime(?date)) >= 2005 ) 
              && ( ?conference IN (dblprc:vldb, dblprc:sigmod) ) )
          }
        GROUP BY ?author
        HAVING ( COUNT(?paper) >= 20 )
      }
    }
  \end{lstlisting}

The dataframe required for this task is extracted from the DBLP knowledge graph represented in RDF through the sequence of \rdfframes operators shown in Listing~\ref{lst:topic}.
First, we identify the authors who have published 20 or more papers in SIGMOD and VLDB since the year 2000,
which requires using the \rdfframes grouping, aggregation, and filtering capabilities.
For the purpose of this case study, these are considered the thought leaders
of the field of databases.
Next, we find the titles of all papers published by these authors since 2010.
The \sparql query generated by \rdfframes is shown in Listing~\ref{lst:topicsparql}.

We then run topic modeling on the titles to identify the topics of the papers,
which we consider to be the active topics of database research.
We use off-the-shelf components from the rich ecosystem of pandas libraries to implement topic modeling (see Appendix~\ref{app:fullcode}). Specifically, we use NLP libraries for stop-word removal and scikit-learn for topic modeling using SVD.
This shows the benefit of using \rdfframes to get data into a pandas dataframe with a few lines of code,
since one can then utilize components from the PyData ecosystem.

\subsubsection{Knowledge Graph Embedding}
\label{subsec:kgembedding}
\vspace*{-10pt}
Knowledge graph embeddings 
are widely used
relational learning models,
and they are state of the art on
benchmark datasets for link prediction and fact classification~\cite{journals/tkde/WangMWG17, wang2018multi}.
The input to these models is a dataframe of triples, i.e., a table of three columns: $[subject, predicate, object]$ where the $object$ is a URI representing an entity (i.e., not a literal).
Currently, knowledge graph embeddings are typically evaluated only on small pre-processed subsets of knowledge graphs like FB15K~\cite{bordes2013transe} and WN18~\cite{bordes2013transe} rather than the full knowledge graphs, and thus, the validity of their performance results has been questioned recently in multiple papers~\cite{pujara2017sparsity, dettmers2018convolutional}. 
Filtering the knowledge graph to contain only entity-to-entity triples and loading the result in a dataframe is a necessary first step in constructing knowledge graph embedding models on full knowledge graphs.
\rdfframes can perform this first step using one line of code as shown in Listing~\ref{lst:kge} (generated \sparql in Listing~\ref{lst:kgesparql}).
With this line of code, the filtering can be performed efficiently in an RDF engine, and
\rdfframes handles issues related to communication with the engine and integrating with PyData. These issues become important, especially when 
dealing with large knowledge graphs where
the resulting dataframe has millions of rows.

\begin{lstlisting}[
aboveskip=0.5\baselineskip,
belowskip=-0.0\baselineskip,
language=Python,
breaklines=true,
showspaces=false,
basicstyle=\ttfamily\scriptsize,
commentstyle=\color{gray},
otherkeywords={feature_domain_range, expand, filter, group_by, join, entities, .count, cache, unique, select_cols},
keywordstyle=\color{blue},
caption={\rdfframes code - Knowledge graph embedding.},
captionpos=b,
label={lst:kge}]
graph.feature_domain_range(s, p, o)\
     .filter({o: ['isURI']})
\end{lstlisting}

\begin{lstlisting}[
aboveskip=-0.0\baselineskip,
belowskip=-0.0\baselineskip,
  language=SQL,
  breaklines=true,
  showspaces=false,
  basicstyle=\ttfamily\scriptsize,
  commentstyle=\color{gray},
  keywordstyle=\color{blue},
  otherkeywords={OPTIONAL, FILTER},
  caption={\sparql query corresponding to \rdfframes code shown in Listing~\ref{lst:kge}.},
  captionpos=b,
  label={lst:kgesparql}
  ]
SELECT * 
FROM <http://dblp.13s.de/>
WHERE {
        ?sub ?pred ?obj .
        FILTER ( isIRI(?obj) ) 
      }
\end{lstlisting}

\vspace*{-6pt}
\subsection{Synthetic Workload}
\label{subsec:sythetic}

\vspace*{-3pt}
While the case studies in the previous section show \rdfframes in real applications,
it is still desirable to have a more comprehensive evaluation of the framework.
To this end, we created a synthetic workload
consisting of 16 queries
written in \rdfframes that exercise different capabilities of the framework.
All the queries are on the DBpedia knowledge graph, and two queries join DBpedia with the YAGO knowledge graph.
One query joins the three knowledge graphs DBpedia, YAGO, and DBLP.
Four of the queries use only expand and filter (up to 10 expands, including some with optional predicates).
Four of the queries use grouping with expand (including one with expand after the grouping).
Eight of the queries use joins, including complex queries that exercise features such as outer join, multiple joins, joins between different graphs, and joins on grouped datasets.
A description of the  queries and the \rdfframes features and \sparql capabilities that they exercise
can be found in Appendix~\ref{app:queries}.

\subsection{Experimental Setup}
\label{subsec:expsetup}

\vspace*{5pt}
\subsubsection{Dataset Details}

\vspace*{-7pt}
The three knowledge graphs used in the evaluation have
different sizes and statistical features.
The first is the English version of the DBpedia knowledge graph. We extracted the December 2020 core collection from DBpedia Databus.\footnote{\url{https://databus.dbpedia.org/dbpedia/collections/latest-core}} The collection contains 6 billion triples.
The second is the DBLP computer science bibliography dataset
(2017 version)
containing 88 million triples.\footnote{\url{http://www.rdfhdt.org/datasets}}
The third (used in three queries in the synthetic workload) is YAGO version 3.1, containing 1.6 billion triples.
DBLP is relatively small, structured and dense, while DBpedia and YAGO are heterogeneous and sparse.

\vspace*{3pt}
\subsubsection{Hardware and Software Configuration}

\vspace*{4pt}
We use an Ubuntu server with 128GB of memory
to
run a Virtuoso OpenLink Server (version 7.2.6-rc1.3230-pthreads as of Jan 9 2019) with its default configuration.
We load the DBpedia, DBLP, and YAGO knowledge graphs to the Virtuoso server. \rdfframes connects to the server to process \sparql queries over HTTP using \texttt{SPARQLWrapper},\footnote{\url{https://rdflib.github.io/sparqlwrapper}}
a Python library that provides
a wrapper for \sparql endpoints.
Recall that the decision to communicate with the server over HTTP rather than the cursor mechanism of Virtuoso was made to ensure maximum generality and flexibility.
When sending \sparql queries directly to the server, we use the \texttt{curl} tool.
The client always runs on a separate core of the same machine as the Virtuoso server so we do not incur communication overhead. In all experiments, we report the average running time of three runs.

\vspace*{3pt}
\subsubsection{Alternatives Compared}

\vspace*{4pt}
Our goal is to evaluate the design decisions of \rdfframes
and to compare it against alternative baselines.
To evaluate the design decisions of \rdfframes,
we ask two questions:
(1)~How important is it to generate optimized \sparql queries rather than using a simple query generation approach?
and
(2)~How important is it to push the processing of relational operators into the RDF engine?
Both of these design choices are clearly beneficial
and the goal is to quantify the benefit.

To answer the first question, we compare \rdfframes against an alternative that 
uses a naive query generation strategy. Specifically, for each API call to \rdfframes, we generate a subquery that contains the pattern corresponding to that API call and we finally join all the subqueries in one level of nesting with one outer query. For example, each call to an expand creates a new subquery containing one triple pattern described by the expand operator.
We refer to this alternative as \texttt{Naive Query Generation}.
The naive queries for the first two case studies are shown in Appendices~\ref{app:naiveQuery} and~\ref{app:naiveQuery2}. The \sparql query for the third case study is simple enough that Listing~\ref{lst:kgesparql} is also the naive query.

To answer the second question, 
we compare to an alternative that uses \rdfframes (with optimized query generation) only for graph navigation using the $seed$ and $expand$ operators,
and performs any relational-style processing in pandas. We refer to this alternative as
\texttt{Navigation + pandas}.

If we do not use \rdfframes, we can envision three alternatives for pre-processing the data and loading it into a dataframe, and we compare against all three:
\squishlist
\item
Do not use an RDF engine at all, but rather write an ad-hoc script that runs on the knowledge graph stored in some RDF serialization format.
To implement this solution we write scripts using the rdflib library\footnote{\url{https://github.com/RDFLib/rdflib}} to load the RDF dataset into pandas, and use pandas operators for any additional processing.
The rdflib library can process any RDF serialization format, and in our case the data was stored in the N-Triples format. We refer to this alternative as \texttt{rdflib + pandas}.

\item
Use an RDF engine, and use a simple \sparql query to load the RDF dataset into a dataframe. Use pandas for additional processing.
This is a variant of the first alternative but it uses \sparql instead of rdflib.
The advantage is that the required \sparql is very simple, but still
benefits from 
the processing capabilities of the RDF engine.
We refer to this alternative as \texttt{\sparql + pandas}.

\item
Use a \sparql query written by an expert (in this case, the authors of the paper) to do all the pre-processing inside the RDF engine and output the result to a dataframe.
This alternative takes full advantage of the capabilities of the RDF engine, but suffers from the ``impedance mismatch'' described in the introduction: \sparql uses a different programming style compared to machine learning tools and requires expertise to write,
and additional code is required to export the data into a dataframe.
We refer to this alternative as \texttt{Expert \sparql}.

\squishend

We verify that the results of all alternatives are identical. Note that \rdfframes, \texttt{Naive Query Generation},
and \texttt{Expert \sparql} generate semantically equivalent \sparql queries.
The query optimizer of an RDF engine should be able to produce query execution plans for all three queries that are identical or at least have similar execution cost. We will see that Virtuoso, being an industry-strength RDF engine, does indeed deliver the same performance for all three queries in many cases. However, we will also see that there are cases where this is not true, which is expected due to the complexity of optimizing \sparql queries.

\vspace*{-6pt}
\subsection{Results on Case Studies}
\label{subsec:results-casestudies}

\begin{figure*}[tb]
  \centering
  \subfloat[(a)]
    {\includegraphics[width=0.32\textwidth]{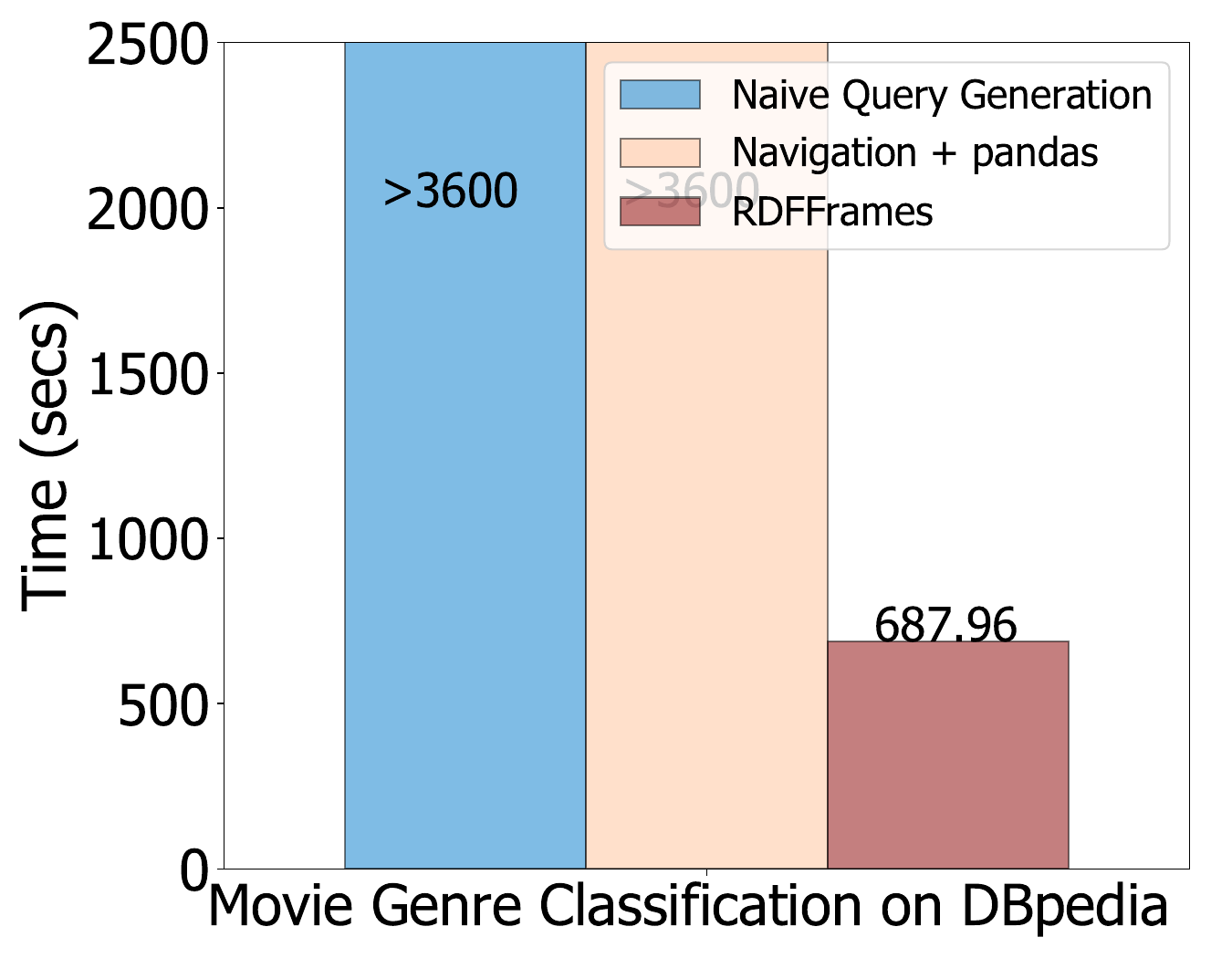}\label{subfig:movie1}}
    \subfloat[(b)]
  {\includegraphics[width=0.32\textwidth]{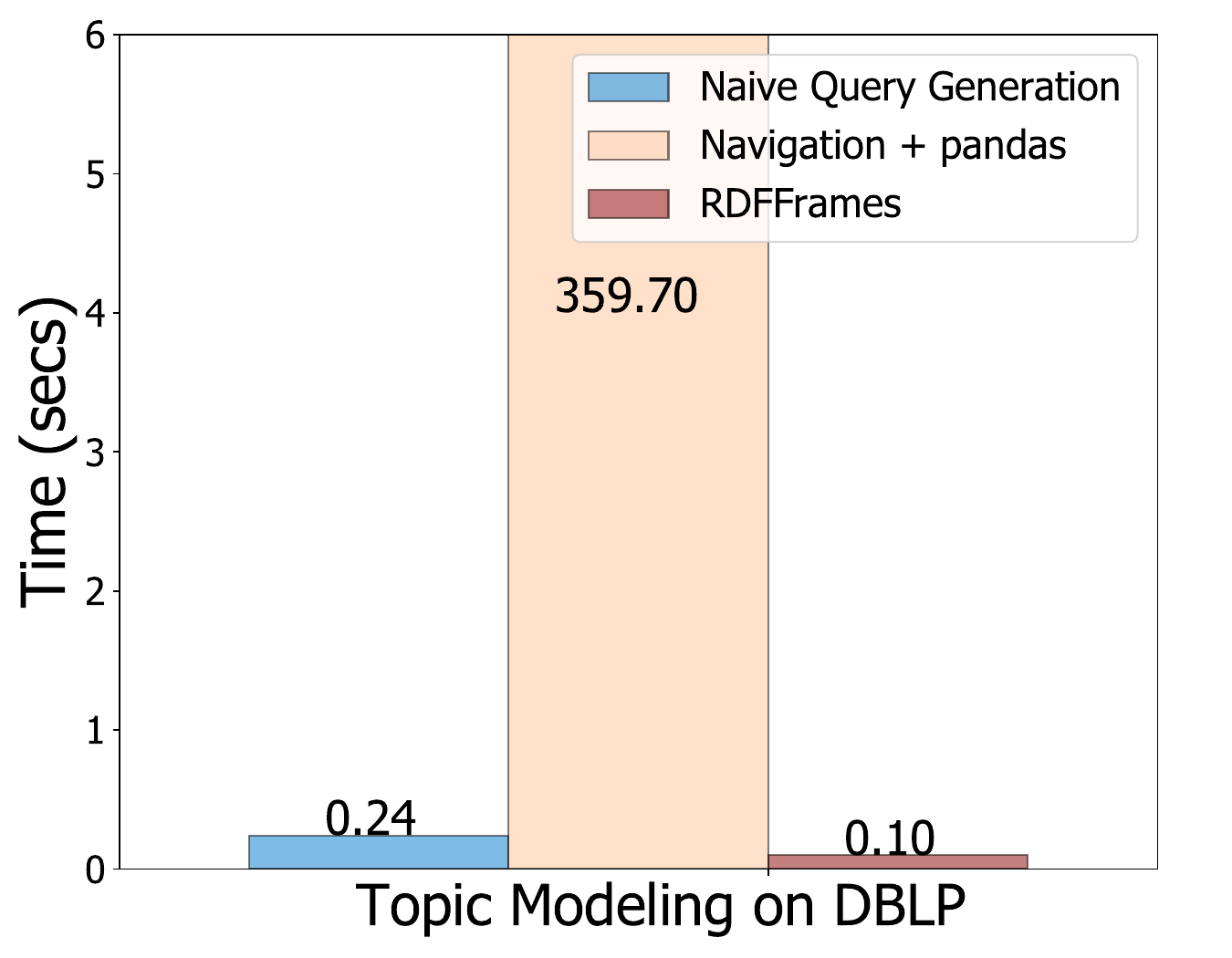}\label{subfig:topic1}}
  \subfloat[(c)]
  {\includegraphics[width=0.32\textwidth]{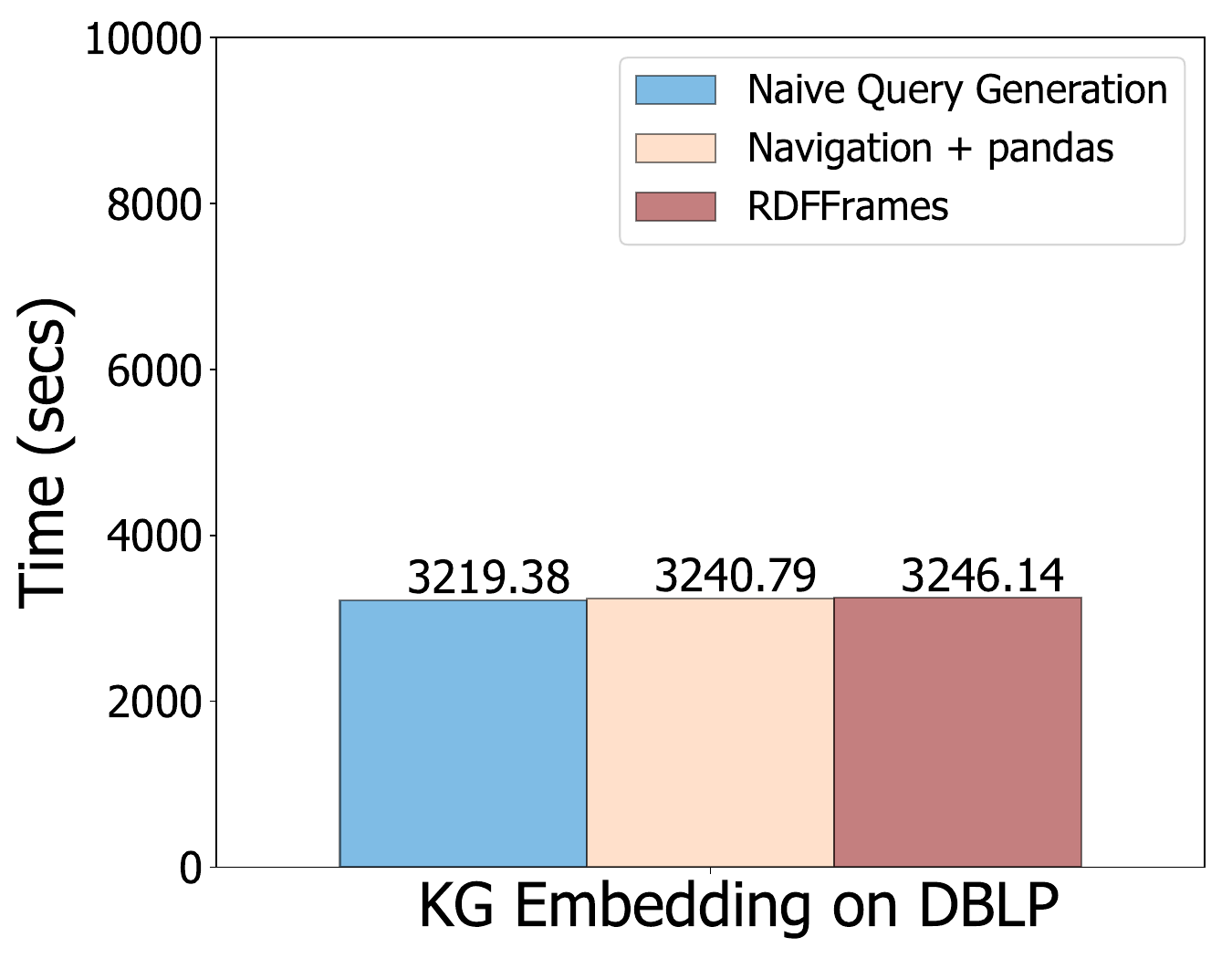}\label{subfig:kge1}}
  \vspace*{-8pt}
  \caption{Evaluating the design of \rdfframes.}
  \label{fig:alternatives}
  \vspace{-\baselineskip}
\end{figure*}

\begin{figure*}[tb]
  \centering
  \subfloat[(a)]
    {\includegraphics[width=0.32\textwidth]{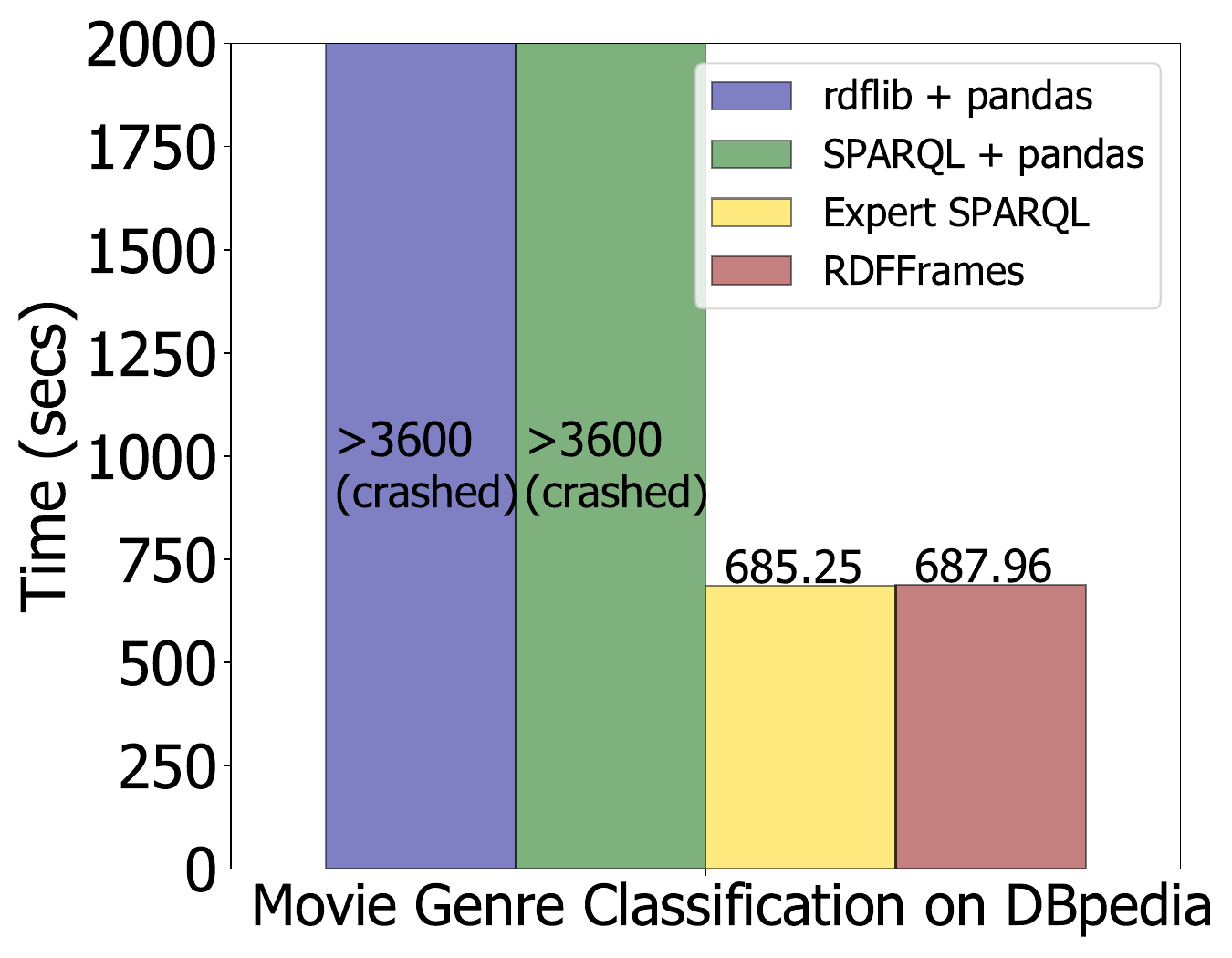}\label{subfig:movie2}}
    \subfloat[(b)]
  {\includegraphics[width=0.32\textwidth]{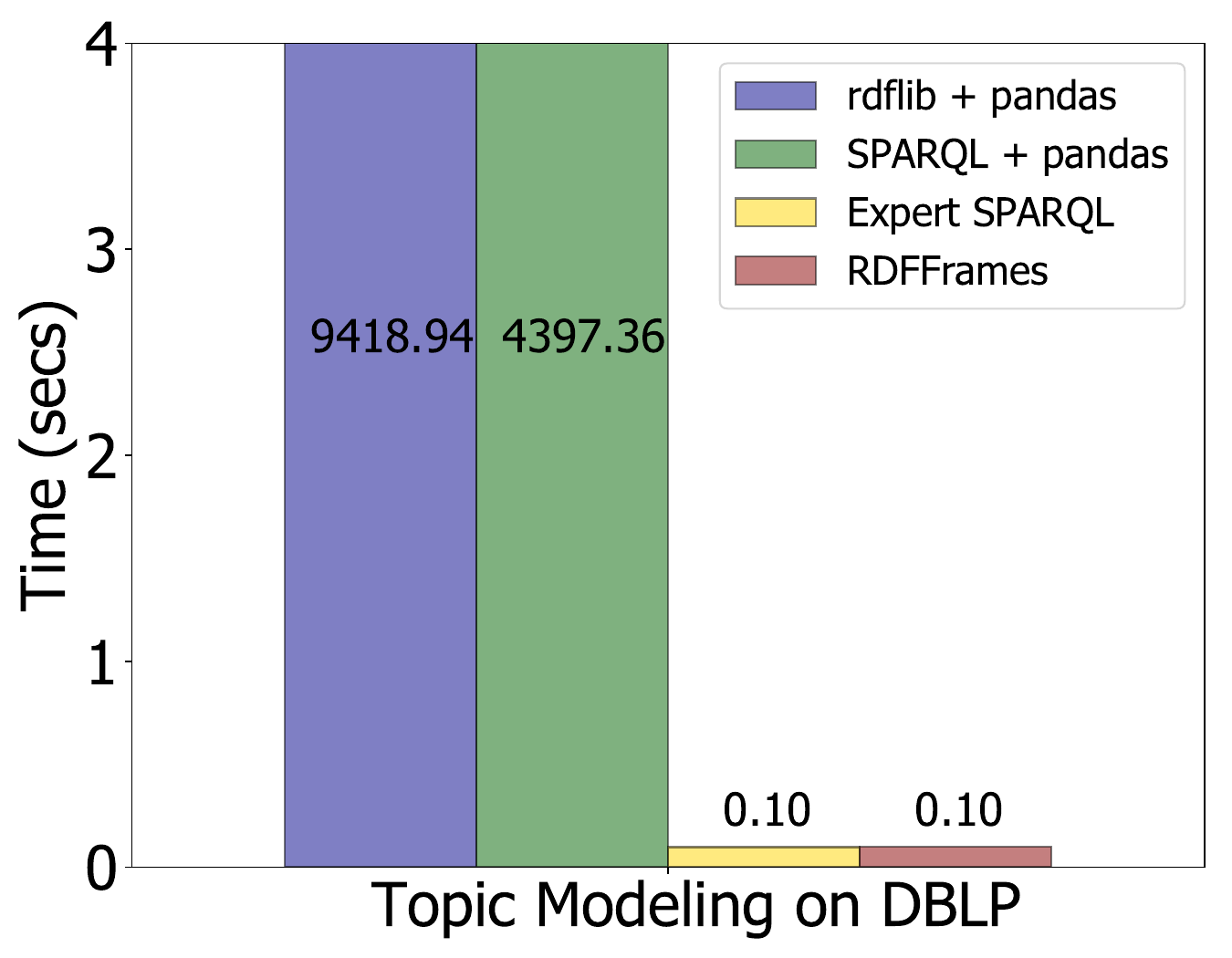}\label{subfig:topic2}}
  \subfloat[(c)]
  {\includegraphics[width=0.32\textwidth]{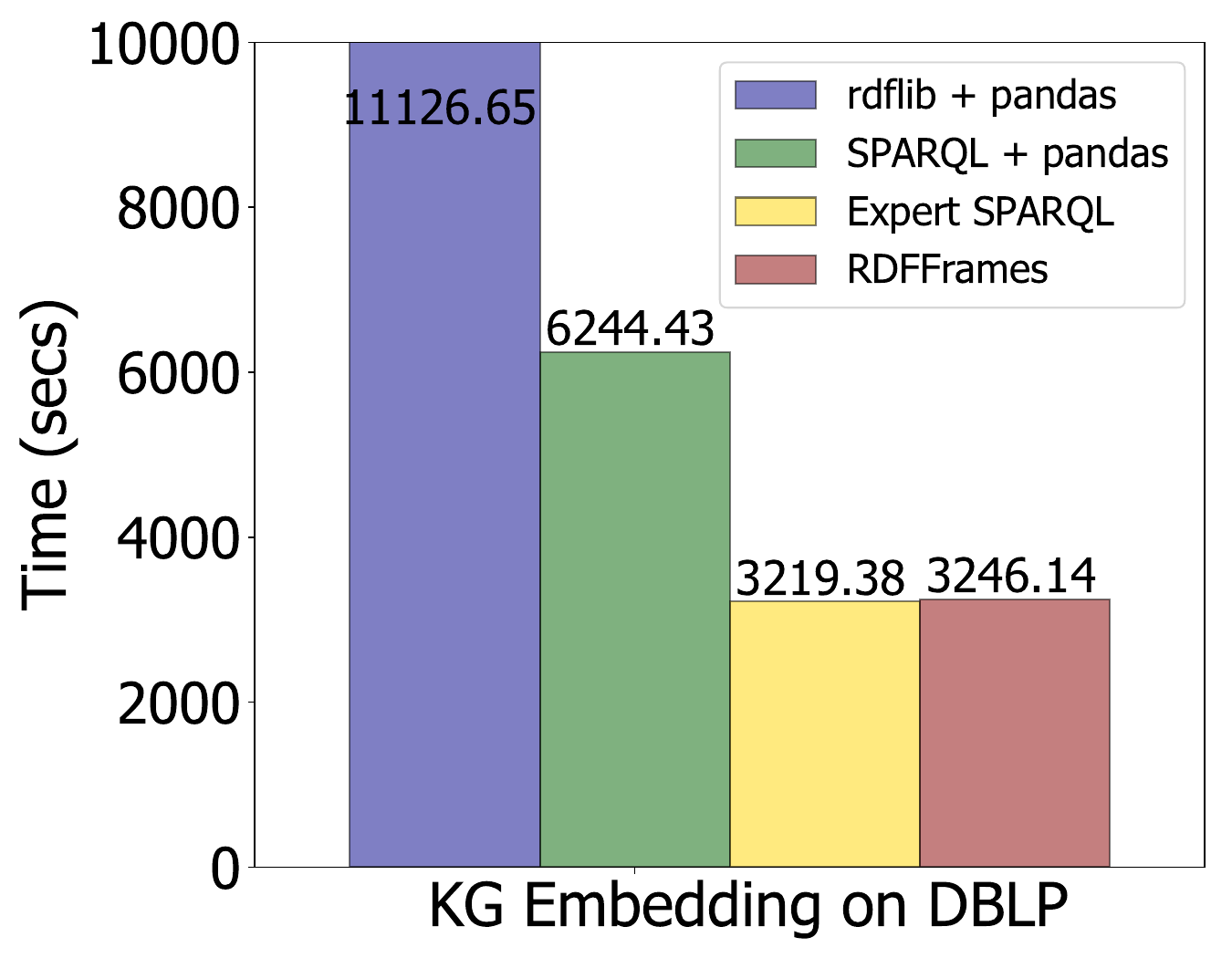}\label{subfig:kge3}}
  \vspace*{-8pt}
  \caption{Comparing \rdfframes to alternative baselines.}
  \label{fig:performance}
  \vspace{-\baselineskip}
\end{figure*}

\vspace*{2pt}
\subsubsection{Evaluating the design decisions of \rdfframes}

\vspace*{-4pt}
Figure~\ref{fig:alternatives} 
shows the running time of 
\texttt{Naive Query Generation}, \texttt{Navigation + pandas},
and \rdfframes on the three case studies.
 
\vspace{4pt}
\noindent
\textbf{Movie Genre Classification.}
This task requires heavy processing 
on the DBpedia dataset and returns a dataframe of 19,633 movies.
The results are presented in Figure~\ref{fig:alternatives}(a).

The running time of \rdfframes is 687.96 seconds. Of this time, less than 5 milliseconds is spent on
preparing the \sparql query (i.e., recording the \rdfframes operations, generating the query model, and producing the query). The remaining time is spent on issuing the query to the engine and retrieving the results.
This is typical in all our experiments: \rdfframes needs a few milliseconds to generate the \sparql query and the remaining time is spent on query processing. 
The query produced by naive query generation did not finish in one hour and we terminated it after this time,
which demonstrates the need for \rdfframes to generate optimized \sparql and not rely exclusively on the query optimizer.
The \texttt{Navigation + pandas} alternative also timed out after one hour,
which demonstrates the need for pushing computation into the engine.

\vspace{4pt}
\noindent
\textbf{Topic Modeling.}
This task requires heavy processing on the DBLP dataset and returns a dataframe of 4,209 titles.
The results are depicted in Figure~\ref{fig:alternatives}(b).
Naive query generation did well here, with the query optimizer generating a good plan for the query Nonetheless, naive query generation is 2x slower than \rdfframes.
This further demonstrates the need for generating optimized \sparql.
The \texttt{Navigation + pandas} alternative here was particularly bad,
reinforcing the need to push computation into the engine.

\vspace{4pt}
\noindent
\textbf{Knowledge Graph Embedding.}
This task keeps only triples where the object is an entity (i.e., not a literal).
It does not require heavy processing but requires handling the scalability issues of returning a huge final dataframe with all triples of interest.
The results on DBLP are shown in Figure~\ref{fig:alternatives}(c).
All the alternatives have similar performance for this task, since the required \sparql query is simple and processed well by Virtuoso, and since there is no processing required in pandas.

\vspace*{-4pt}
\subsubsection{Comparing \rdfframes With Alternative Baselines}

\vspace*{-10pt}
Figure~\ref{fig:performance} compares the running time of \rdfframes on the three case studies
to the three alternative baselines:
\texttt{rdflib + pandas},
\texttt{\sparql + pandas},
and
\texttt{Expert \sparql}.

\vspace{4pt}
\noindent
\textbf{Movie Genre Classification.}
Both the \texttt{rdflib + pandas} and \texttt{\sparql + pandas}
baselines crashed after 
more than one hour
due to scalability issues, showing that they are not viable alternatives.
On the other hand, \rdfframes and \texttt{Expert \sparql} have similar performance.
This shows that \rdfframes does not add overhead and is able
to match the performance of an expert-written \sparql query,
which is the best case for an automatic query generator.
Thus, the flexibility and usability of \rdfframes does not come at the cost of reduced performance.

\vspace{4pt}
\noindent
\textbf{Topic Modeling.}
The baselines that perform computation in pandas did not crash as before, but are orders of magnitude slower than \rdfframes and \texttt{Expert \sparql}.
In this case as well, the running time of \rdfframes 
matches the expert-written \sparql query.

\vspace{4pt}
\noindent
\textbf{Knowledge Graph Embedding.}
In this experiment, \texttt{rdflib + pandas}  is 
3x slower than \rdfframes
and \texttt{\sparql + pandas} is 2x slower than \rdfframes, while 
\rdfframes has the same performance as \texttt{Expert \sparql}.
These results reinforce the conclusions drawn earlier.

\begin{figure}
  \centering
    \includegraphics[width=1\columnwidth]{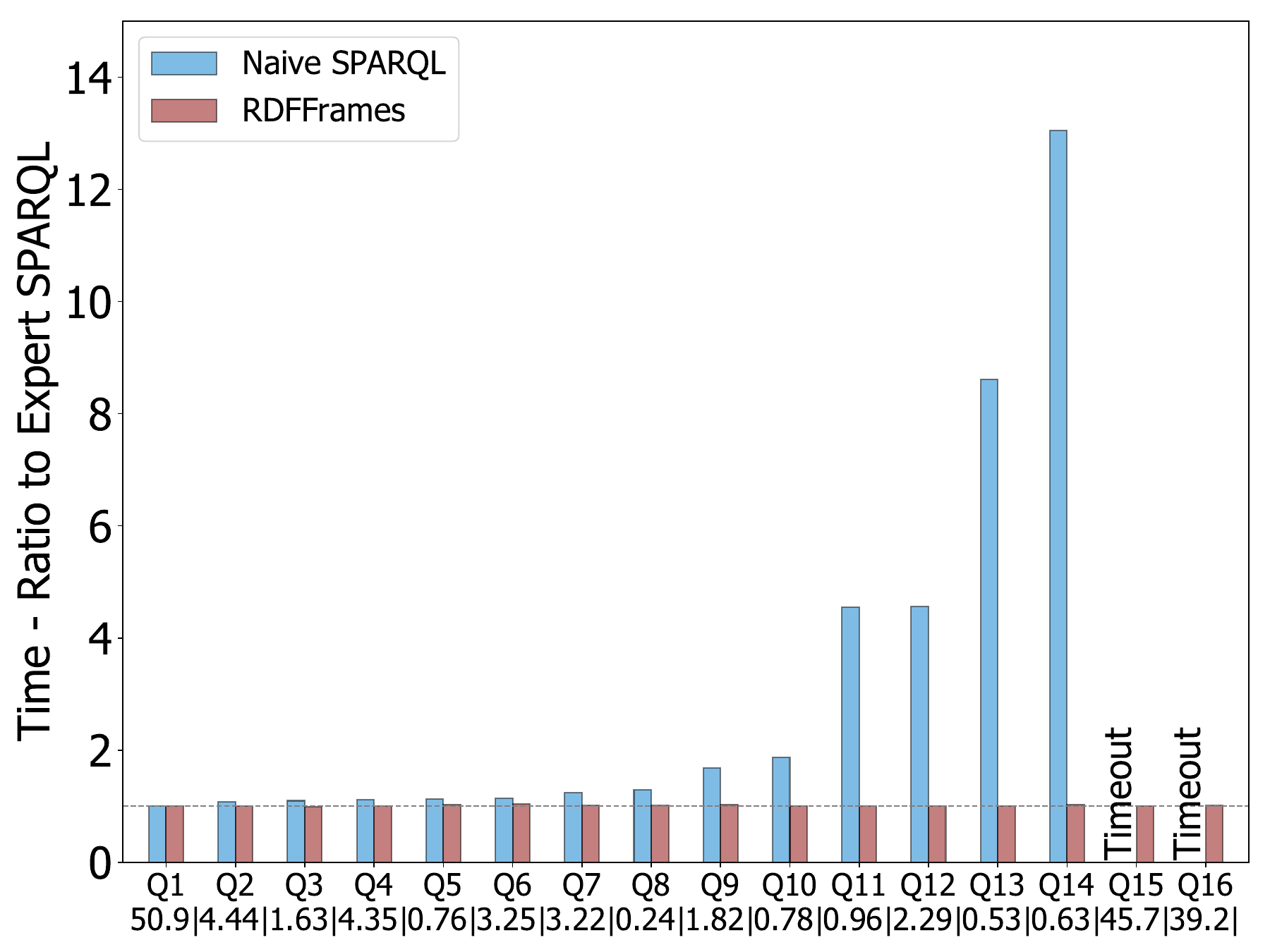}
  \vspace*{-8pt}
  \caption{Results on the synthetic workload.}
  \label{fig:synthetic}
\end{figure}

\subsection{Results on Synthetic Workload}
\label{subsec:results-synthetic}

In this experiment, we use the synthetic workload of 16 queries to do a more comprehensive evaluation of \rdfframes.
The previous section showed that 
\texttt{Navigation + pandas},
\texttt{rdflib + pandas}, and
\texttt{\sparql + pandas} are not competitive with \rdfframes.
Thus, we exclude them from this experiment.
Instead, we focus on the quality of the queries generated by \rdfframes and whether this broad set of queries shows that naive query generation would work well.
Figure~\ref{fig:synthetic} compares naive query generation and \rdfframes to expert-written \sparql.

The y-axis shows the ratio between the running time of naive query generation and expert-written \sparql, and between the running time of \rdfframes and expert-written \sparql.
Thus, expert-written \sparql is considered the gold standard and the figure shows how well the two query generation alternatives match this standard.
To improve the comparison,
the absolute running time of expert-written \sparql in seconds is shown under each query on the x-axis.
The queries are sorted in ascending order by the ratio of
naive query generation to expert-written \sparql
(the blue bars in the figure).
The dashed horizontal line represents a ratio of 1.

The ratios for \rdfframes range between 0.99 and 1.04,
which shows that \rdfframes is good at generating queries that match the performance of expert-written queries. On the other hand, the ratios for naive query generation vary widely. The first six queries have ratios between 1.01 and 1.14.
For these queries, the Virtuoso optimizer does a good job of generating a close-to-optimal plan for the naive query.
The next six queries have ratios between 1.24 and 4.56.
Here, we are seeing the weakness of naive query generation and the need for the optimizations performed by \rdfframes during query generation. 
The situation is worse for the last four queries:  naive query generation is an order of magnitude slower for two queries and the last two queries time out after one hour.
Thus, the results on this more comprehensive workload validate the quality of the queries generated by \rdfframes and the need for its sophisticated query generation algorithm.

\begin{table*}[th!]
  \caption{Running time with varying operator complexity.}
  \label{tab:cost}
  \centering
  \begin{tabular}{|l|c|c|c|c|c|} 
      \hline
   \multicolumn{2}{|c}{} & \multicolumn{4}{|c|}{\textbf{Running Time (seconds)}} \\ \hline
      \textbf{filter (on genre)} & \textbf{No.\ of Movies} & \textbf{count} & \textbf{select} & \textbf{group\_by} & \textbf{join}\\
      \hline
      Sitcom & 5015 & 0.088 & 0.114 & 0.115 & 0.342 \\ \hline 
      Sitcom, Drama, Comedy & 12115 & 0.574 & 0.635 & 0.704 &  1.838 \\ \hline 
      No filter (all genres) & 87811 & 3.26 & 3.163 & 6.286 & 77.896 \\ \hline 
  \end{tabular}
\end{table*}

  \subsection{Effect of Operator Complexity}

\label{subsec:results-complexity}
\vspace*{-3pt}

In our final experiment, we study how the complexity of various \rdfframes operators affects performance.
Unlike the previous experiment, in which we varied the complexity of large queries as indicated in Appendix~\ref{app:queries},
this experiment studies the issue of complexity at the granularity of an operator.
The cost of operators is highly dependent on the RDF engine being used
(Virtuoso in our case).
Nevertheless, we want to see if there are any patterns in performance.

To study the effect of operator complexity, we create four \rdfframes queries of increasing complexity that operate on movies in the DBpedia knowledge graph. 
The performance of these queries is shown in Table~\ref{tab:cost}.
The first query counts the number of movies. We find movies by finding entities that are the subject of a `starring' predicate. We then use the
$expand$ \rdfframes operator to get the movie titles and apply the $count$ aggregation function on these titles.
The second query uses the $select$ \rdfframes operator to retrieve all movie titles.
The third query uses the $group\_by$ operator to group movies by genre and counts the number of movies in each genre.
The fourth query is a join query that finds actors who are also movie directors. This query creates a dataset of actors and a dataset of directors, and then joins the two datasets through an inner $join$ operator on actor/director name.

We ran the four queries with $filter$ operators of varying selectivity.
In one case, we had no filter (i.e.,~we ran the query on all movies).
In the second case, we had a filter specifying that movie genre has to be one of sitcom, drama, or comedy.
In the third case, the filter specified that movie genre has to be sitcom (the most selective filter).

Each row in Table~\ref{tab:cost} represents a filter, and the number of movies retrieved by this filter is presented in the second column. The rows of the table are sorted by this column. The next four columns in each row show the running time of the four queries for this filter. In all cases, the time that \rdfframes spends to generate the \sparql query is less than one millisecond, so we do not report it separately.

Looking at each running time column from top to bottom, we see that the bigger the input data, the more time is required. Looking at each row from left to right, we see that the more complex the query, the more time is required. Both of these results are expected. Another observation about Table~\ref{tab:cost} is that the variation in running times is not excessive. Even the most expensive query, the join query with no filter, which joins a dataset of 31221 actors with a dataset of 1784 directors, runs in a reasonable 77.896 seconds.

Thus, the experiment shows that operator complexity and dataset size do have an effect on performance. The observed effect is in-line with expectations and does not affect the usability of \rdfframes. The robust performance we observe is partly due to the robustness of Virtuoso
and partly due to our process for \sparql query generation.

\section{Conclusion}
\label{sec:conclusion}

We presented \rdfframes, a framework for seamlessly integrating knowledge graphs into machine learning applications.
\rdfframes is based on a number of powerful operators for graph navigation and relational processing that enable users to generate tabular data sets from knowledge graphs using 
procedural programming idioms
that are familiar in machine learning environments such as PyData.
\rdfframes automatically converts these procedural calls to optimized \sparql queries and manages the execution of these queries on a local RDF engine or a remote \sparql endpoint,
shielding the user from all details of \sparql query execution.
We provide a Python implementation of \rdfframes that is tightly integrated with the pandas library
and experimentally demonstrate its efficiency.

Directions for future work include conducting a comprehensive user study to identify and resolve any usability-related issues that could be faced by \rdfframes users. 
A big problem in RDF is that users need to know the knowledge graph vocabulary and structure in order to effectively query it. To address this problem, one direction for future work is expanding the exploration operators of \rdfframes to include keyword search.
Testing and evaluating \rdfframes on multiple RDF engines is another possible future direction.

\bibliographystyle{ACM-Reference-Format}
\bibliography{main}

\pagebreak
\onecolumn
\begin{appendix}
  \section{Full Python Code for Case Studies}
\label{app:fullcode}
\subsection{Movie Genre Classification}
\label{app:moviecase}

\begin{lstlisting}[
  aboveskip=-0.0\baselineskip,
  belowskip=-0.0\baselineskip,
  language=Python,
  showspaces=false,
  basicstyle=\ttfamily\scriptsize,
  commentstyle=\color{gray}, 
  otherkeywords={expand, filter,execute, PANDAS_DF,group_by, join, feature_domain_range, .count, cache, unique},
  keywordstyle=\color{blue},
  caption={Full code for movie genre classification.},
  captionpos=b,
  label={lst:movidecasestudy}]

  # RDFFrames imports, graph, and prefixes
  from rdfframes.knowledge_graph import KnowledgeGraph
  from rdfframes.dataset.rdfpredicate import RDFPredicate
  from rdfframes.utils.constants import JoinType
  from rdfframes.client.http_client import HttpClientDataFormat, HttpClient
  graph = KnowledgeGraph(graph_uri='http://dbpedia.org',
                         prefixes= {'dcterms': 'http://purl.org/dc/terms/',
                                  'rdfs': 'http://www.w3.org/2000/01/rdf-schema#',
                                  'dbpprop': 'http://dbpedia.org/property/',
                                  'dbpr': 'http://dbpedia.org/resource/'}) 
  
  # RDFFrames code for creating the dataframe
  dataset = graph.feature_domain_range('dbpp:starring','movie', 'actor')
  dataset = dataset.expand('actor',[('dbpp:birthPlace', 'actor_country'),('rdfs:label', 'actor_name')])
           .expand('movie', [('rdfs:label', 'movie_name'),('dcterms:subject', 'subject'),
            ('dbpp:country', 'movie_country'),('dbpo:genre', 'genre', Optional)]).cache()
  american = dataset.filter({'actor_country':['regex(str(?actor_country),"USA")']})
  prolific = dataset.group_by(['actor']).count('movie', 'movie_count', unique=True).filter({'movie_count': ['>=100']})
  movies = american.join(prolific,'actor', OuterJoin).join(dataset, 'actor', InnerJoin)
  
  # Client and execution
  output_format = HttpClientDataFormat.PANDAS_DF
  client = HttpClient(endpoint_url=endpoint, return_format=output_format)
  df = movies.execute(client, return_format=output_format)

  # Preprocessing and preparation
  import re
  import nltk
  
  def clean(dataframe):
    for i, row in df.iterrows():
      if df.loc[i]['genre'] != None:
      value =df.at[i, 'genre']
      if re.match(regex,str(value)) is not None:
          df.at[i, 'genre'] = value.split('/')[-1]
    return dataframe

  # Remove URL from the 'genre' and convert to label keys
  df=clean(df)

  # Find the most most frequent genres
  all_genres = nltk.FreqDist(df['genre'].values)
  all_genres_df = pd.DataFrame({'genre':list(all_genres.keys()), 'Count':list(all_genres.values())})
  all_genres_df.sort_values(by=['Count'],ascending=False)
  
  # In this example, use 900 movies as a cut off for the frequent movies
  most_frequent_genres = all_genres_df[all_genres_df['Count']> 900]
  df = df[df['genre'].isin(list(most_frequent_genres['genre']))]

  # Features and factorization
  from sklearn.model_selection import train_test_split
  from sklearn.preprocessing import StandardScaler

  df= df.apply(lambda col: pd.factorize(col, sort=True)[0])
  features = ["movie_name", "actor_name", "actor_country","subject","movie_country", "subject"]
  df = df.dropna(subset=['genre'])
  x = df[features]
  y = df['genre']
  x_train, x_test, y_train, y_test = train_test_split(x, y, random_state=20)
  sc = StandardScaler()
  x_train = sc.fit_transform(x_train)
  x_test = sc.fit_transform(x_test)

  # Random Forest classifier
  from sklearn.ensemble import RandomForestClassifier

  model=RandomForestClassifier(n_estimators=100)
  model.fit(x_train,y_train)
  model.fit(x_train,y_train)
  y_pred=clf.predict(x_test)
  print("Accuracy:",metrics.accuracy_score(y_test, y_pred))
    
\end{lstlisting}
 \pagebreak
\subsection{Topic Modeling}
\label{app:topiccase}

\begin{lstlisting}[
  aboveskip=0.5\baselineskip,
  belowskip=-0.0\baselineskip,
  language=Python,
  showspaces=false,
  basicstyle=\ttfamily\scriptsize,
  commentstyle=\color{gray},
  otherkeywords={KnowledgeGraph,PANDAS_DF, execute, rdfframes,expand, filter, group_by, join, entities, count, cache, unique, select_cols,apply, lower, nltk,sklearn,TfidfVectorizer,TruncatedSVD,fit,get_feature_names,lambda},
  keywordstyle=\color{blue},
  caption={Full code for topic modeling.},
  captionpos=b,
  label={lst:topiccasestudy}]
  # RDFFrames imports, graph, prefixes, and client
  import pandas as pd
  from rdfframes.client.http_client import HttpClientDataFormat, HttpClient
  from rdfframes.knowledge_graph import KnowledgeGraph
  graph = KnowledgeGraph(
          graph_uri = 'http://dblp.l3s.de',
          prefixes = {"xsd": "http://www.w3.org/2001/XMLSchema#",
          "swrc": "http://swrc.ontoware.org/ontology#",
          "rdf": "http://www.w3.org/1999/02/22-rdf-syntax-ns#",
          "dc": "http://purl.org/dc/elements/1.1/",
          "dcterm": "http://purl.org/dc/terms/",
          "dblprc": "http://dblp.l3s.de/d2r/resource/conferences/"
      })

  output_format = HttpClientDataFormat.PANDAS_DF
  client = HttpClient(endpoint_url=endpoint, port=port,return_format=output_format)

  # RDFFrames code for creating the dataframe
  papers = graph.entities('swrc:InProceedings', paper)
  papers = papers.expand('paper',[('dc:creator', 'author'),('dcterm:issued', 'date'), ('swrc:series', 'conference'),
                        ('dc:title', 'title')]).cache()
  authors = papers.filter({'date': ['>=2005'],'conference': ['In(dblp:vldb, dblp:sigmod)']}).group_by(['author'])
                   . count('paper', 'n_papers').filter({'n_papers': '>=20', 'date': ['>=2005']})
  titles = papers.join(authors, 'author', InnerJoin).select_cols(['title'])
  df = titles.execute(client, return_format=output_format)

  # Preprocessing and cleaning
  from nltk.corpus import stopwords
  df['clean_title'] = df['title'].str.replace("[^a-zA-Z#]", " ")
  df['clean_title'] = df['clean_title'].apply(lambda x: x.lower())
  df['clean_title'] = df['clean_title'].apply(lambda x: ' '.join([w for w in str(x).split() if len(w)>3])) 
  stop_words        = stopwords.words('english')
  tokenized_doc     = df['clean_title'].apply(lambda x: x.split())
  df['clean_title'] = tokenized_doc.apply(lambda x:[item for item in x if item not in stop_words])

  # Vectorization and SVD model using the scikit-learn library
  from sklearn.feature_extraction.text import TfidfVectorizer
  from sklearn.decomposition import TruncatedSVD
  vectorizer  = TfidfVectorizer(stop_words='english', max_features= 1000, max_df = 0.5, smooth_idf=True)
  Tfidf_title = vectorizer.fit_transform(df['clean_title'])
  svd_model   = TruncatedSVD(n_components=20, algorithm='randomized',n_iter=100, random_state=122)
  svd_model.fit(Tfidf_titles)  

  # Extracting the learned topics and their keyterms
  terms = vectorizer.get_feature_names()
  for i, comp in enumerate(svd_model.components_):
      terms_comp   = zip(terms, comp)
      sorted_terms = sorted(terms_comp, key= lambda x:x[1], reverse=True)[:7]
      print_string = "Topic"+str(i)+": "
      for t in sorted_terms:
          print_string += t[0] + " "
  \end{lstlisting}

 \pagebreak
\subsection{Knowledge Graph Embedding}
\label{app:kgecase}

\vspace*{8pt}
  \begin{lstlisting}[
    aboveskip=-1.0\baselineskip,
    belowskip=-0.0\baselineskip,
    language=Python,
    showspaces=false,
    basicstyle=\ttfamily\scriptsize,
    commentstyle=\color{gray},
    otherkeywords={KnowledgeGraph, execute, rdfframes, expand, filter, join, group_by, feature_domain_range, execute, PANDAS_DF, ampligraph, train_test_split_no_unseen, ComplEx, fit,mr_score, mrr_score, evaluate_performance, expand, filter, group_by, join, entities, count, cache, unique, select_cols},
    keywordstyle=\color{blue},
    caption={Full code for knowledge graph embedding.},
    captionpos=b,
    label={lst:kgecasestudy}]
    # Get all triples where the object is a URI
  from rdfframes.knowledge_graph import KnowledgeGraph
  from rdfframes.dataset.rdfpredicate import RDFPredicate
  from rdfframes.client.http_client import HttpClientDataFormat, HttpClient    
  output_format = HttpClientDataFormat.PANDAS_DF
  client = HttpClient(endpoint_url=endpoint,
                        port=port,
                        return_format=output_format,
                        timeout=timeout,
                        default_graph_uri=default_graph_url,
                        max_rows=max_rows
                        )
  dataset = graph.feature_domain_range(s, p, o).filter({o: ['isURI']})
  df = dataset.execute(client, return_format=output_format)
    
  # Train/test split and create ComplEx model from ampligraph library
  from ampligraph.evaluation import train_test_split_no_unseen 
  triples = df.to_numpy()
  X_train, X_test = train_test_split_no_unseen(triples, test_size=10000)
  #complEx model from ampligraph library
  from ampligraph.latent_features import ComplEx
  from ampligraph.evaluation import evaluate_performance, mrr_score, hits_at_n_score
  model = ComplEx(batches_count=50,epochs=300,k=100,eta=20, optimizer='adam',optimizer_params={'lr':1e-4}, 
          loss='multiclass_nll',regularizer='LP', regularizer_params={'p':3, 'lambda':1e-5}, seed=0,verbose=True)
  model.fit(X_train)
  # Evaluate embedding model
  filter_triples = np.concatenate((X_train, X_test))
  ranks = evaluate_performance(X_test, model=model, filter_triples=filter_triples,
                                use_default_protocol=True, verbose=True)
  mr  = mr_score(ranks)
  mrr = mrr_score(ranks)
  \end{lstlisting}

\pagebreak

\section{Description of Queries in the Synthetic Workload}
\label{app:queries}
\vspace*{-3pt}
\hfill
\begin{table*}[th!]
  \caption{Description of the queries in the sytnthetic workload.\label{tab:def-queries}}
  \resizebox{\textwidth}{!}{%
  \begin{tabular}{|c|p{10cm}|p{4cm}|p{3cm}|} 
    \hline
    \textbf{Query} & \textbf{English Description} & \textbf{\rdfframes Operators}& \textbf{\sparql Features}\\
    \hline 
    Q1 & Get a list of films in DBpedia. For each film, return the actor, language, country, genre, story, and studio, in addition to the director, producer, and title (if available). & expand (including optional predicates) & OPTIONAL, DISTINCT\\
    \hline

    Q2 & Get a list of actors available in the DBpedia or YAGO graphs. & join (outer) between two graphs, filter & OPTIONAL, FILTER, UNION
    \\
    \hline
    Q3 & Get a list of American actors available in both the DBpedia and YAGO graphs. & join (inner) between two graphs, expand, filter &  FILTER
    \\
    \hline 
    Q4 & 
    Get the nationality, place of birth, and date of birth of each basketball player in DBpedia, in addition to the sponsor, name, and president of his team (if available). & join (left outer) between two expandable datasets, expand (including optional predicates) & OPTIONAL
    \\
    \hline 
    Q5 &  Get the players (athletes) in DBpedia and their teams, group by teams, count players, and expand the team's name. &  group\_by, count, expand & GROUP BY, COUNT, DISTINCT
    \\
    \hline
    Q6 &  For films in DBpedia that are produced by any studio in India or the United States excluding 'Eskay Movies', and that have one of the following genres (film score, soundtrack, rock music, house music, or dubstep),
    get the actor, director, producer, time and language. & expand, filter & FILTER
    \\ 
    \hline
    Q7 & For the films in DBpedia, get actors, director, country, producer, language, title, genre, story, and studio. Filter on country, studio, genre, and runtime. &  expand, filter & FILTER
    \\
    \hline
    Q8 & Get the nationality, place of birth, and date of birth of each basketball player in DBpedia, in addition to the sponsor, name, and president of his team. & join (inner) between two expandable datasets, expand & Multiple conjunctive graph patterns
    \\
    \hline 
    Q9 & Get the list of basketball players in DBpedia, their teams, and the number of players on each team. &  group\_by, count, expand & GROUP BY, COUNT, DISTINCT
    \\
    \hline
    Q10 &  For films in DBpedia that are produced by any studio in India or the United States excluding 'Eskay Movies', and that have one of the following genres (film score, soundtrack, rock music, house music, or dubstep),
    get the actor and language, in addition to the producer, director, and title (if available). & expand (including optional predicates), filter & OPTIONAL, FILTER
    \\ 
    \hline
    Q11 & Get the list of athletes in DBpedia. For each athlete, return his birthplace and the number of athletes who were born in that place. &  group\_by, count, expand & GROUP BY, COUNT, DISTINCT
    \\ 
    \hline
    Q12 & Get the pairs of films in DBpedia that belong to the same genre and are produced in the same country. For each film in each pair, return the actor, country, story, language, genre, and studio, in addition to the director, producer, and title (if available). & group\_by on multiple columns, count, expand (including optional predicates) &  GROUP BY, COUNT, OPTIONAL, DISTINCT
    \\
    \hline
    Q13 & Get the sponsor, name, president, and the number of basketball players of each basketball team in DBpedia. & join (inner) between two
    datasests (expandable, group\_by) & GROUP BY, COUNT, DISTINCT
    \\
    \hline 
    Q14 & Get the sponsor, name, president, and the number of basketball players (if available) of each basketball team in DBpedia. & join (left outer) between two datasets (expandable, group\_by), expand (including optional predicates) & GROUP BY, COUNT, OPTIONAL, DISTINCT
    \\
    \hline 
    Q15 & Get a list of the books in DBpedia that were written by American authors who wrote more than two books. For each author, return the birth place, country, and education, and for each book return the title, subject, country (if available), and publisher (if available). &  join (outer), group\_by, count, expand (including optional predicates), filter & GROUP BY, COUNT, HAVING,
    OPTIONAL, FILTER, UNION
    \\
    \hline

    Q16 & Get a list of people in the DBpedia graph who were born in the United States. Get a list of authors from the DBLP graph who have publications dated after 2015. Get a list of people in the YAGO graph who are citizens of the United States. Join the three lists retrieved from the three graphs on name. & join (full outer) between three graphs, expand (including optional predicates), filter & OPTIONAL, FILTER, DISTINCT, UNION
    \\
    \hline
  
  \end{tabular}
  }
\end{table*}

\pagebreak
\section{Naive \sparql Query for Movie Genre Classification}
\label{app:naiveQuery}
\vspace*{-2pt}

\begin{lstlisting}[
  breaklines=true,
  columns=fullflexible,
  aboveskip=-0.0\baselineskip,
  belowskip=-0.0\baselineskip,
  language=SQL,
  showspaces=false,
  basicstyle=\ttfamily\scriptsize,
  commentstyle=\color{gray},
  keywordstyle=\color{blue},
  otherkeywords={OPTIONAL, FILTER},
  caption={Naive \sparql query corresponding to the \sparql query shown in Listing~\ref{lst:genresparql}.},
  captionpos=b,
  label={lst:naivemovie}
  ]
  PREFIX  dbpp: <http://dbpedia.org/property/>
  PREFIX  dcterms: <http://purl.org/dc/terms/>
  PREFIX  rdfs: <http://www.w3.org/2000/01/rdf-schema#>
  PREFIX  dbpo: <http://dbpedia.org/ontology/>
  PREFIX  dbpr: <http://dbpedia.org/resource/> 
  SELECT DISTINCT  ?actor_name ?movie_name ?actor_country ?subject ?genre
  FROM <http://dbpedia.org> WHERE
    {{{ SELECT  * WHERE
              {{ SELECT  * WHERE
                  { { SELECT  ?movie ?actor WHERE
                        { ?movie  dbpp:starring  ?actor } }
                    { SELECT  ?actor ?actor_country WHERE
                        { ?actor  dbpp:birthPlace  ?actor_country } }
                    { SELECT  ?actor ?actor_name WHERE
                      { ?actor  rdfs:label  ?actor_name } }
                    { SELECT  ?movie ?movie_name WHERE
                      { ?movie  rdfs:label  ?movie_name } }
                    { SELECT  ?movie ?subject WHERE
                        { ?movie  dcterms:subject  ?subject } }
                    { SELECT  ?movie ?movie_country WHERE
                         { ?movie  dbpp:country  ?movie_country } }
                    { SELECT  ?actor ?actor_country WHERE
                        { ?actor  dbpp:birthPlace  ?actor_country
                          FILTER regex(str(?actor_country), "USA") } }
                    { SELECT  ?movie ?genre WHERE
                        { OPTIONAL
                            { ?movie  dbpo:genre  ?genre } } } } }
              OPTIONAL
                { SELECT DISTINCT  ?actor (COUNT(DISTINCT ?movie) AS ?movie_count) WHERE
                    { { SELECT  ?movie ?actor WHERE
                          { ?movie  dbpp:starring  ?actor } }
                      { SELECT  ?actor ?actor_country WHERE
                          { ?actor  dbpp:birthPlace  ?actor_country } }
                      { SELECT  ?actor ?actor_name WHERE
                          { ?actor  rdfs:label  ?actor_name } }
                      { SELECT  ?movie ?movie_name WHERE
                          { ?movie  rdfs:label  ?movie_name } }
                      { SELECT  ?movie ?subject WHERE
                          { ?movie  dcterms:subject  ?subject } }
                      { SELECT  ?movie ?movie_country WHERE
                          { ?movie  dbpp:country  ?movie_country } }
                      { SELECT  ?movie ?genre WHERE
                          { OPTIONAL
                              { ?movie  dbpo:genre  ?genre } } } }
                  GROUP BY ?actor
                  HAVING ( COUNT(DISTINCT ?movie) >= 100 ) } } }
        UNION
          { SELECT  * WHERE
            { { SELECT DISTINCT  ?actor (COUNT(DISTINCT ?movie) AS ?movie_count) WHERE
                  { { SELECT  ?movie ?actor WHERE
                        { ?movie  dbpp:starring  ?actor } }
                    { SELECT  ?actor ?actor_country WHERE
                        { ?actor  dbpp:birthPlace  ?actor_country } }
                    { SELECT  ?actor ?actor_name WHERE
                        { ?actor  rdfs:label  ?actor_name } }
                    { SELECT  ?movie ?movie_name WHERE
                        { ?movie  rdfs:label  ?movie_name } }
                    { SELECT  ?movie ?subject WHERE
                        { ?movie  dcterms:subject  ?subject } }
                    { SELECT  ?movie ?movie_country WHERE
                        { ?movie  dbpp:country  ?movie_country } }
                    { SELECT  ?movie ?genre WHERE
                        { OPTIONAL { ?movie  dbpo:genre  ?genre } } } }
                  GROUP BY ?actor
                  HAVING ( COUNT(DISTINCT ?movie) >= 100 ) }
                OPTIONAL
                { SELECT  * WHERE
                    { { SELECT  ?movie ?actor WHERE
                          { ?movie  dbpp:starring  ?actor } }
                      { SELECT  ?actor ?actor_country WHERE
                          { ?actor  dbpp:birthPlace  ?actor_country } }
                      { SELECT  ?actor ?actor_name WHERE
                          { ?actor  rdfs:label  ?actor_name } }
                      { SELECT  ?movie ?movie_name WHERE
                          { ?movie  rdfs:label  ?movie_name } }
                      { SELECT  ?movie ?subject WHERE
                          { ?movie  dcterms:subject  ?subject } }
                      { SELECT  ?movie ?movie_country WHERE
                        { ?movie  dbpp:country  ?movie_country } }
                      { SELECT  ?actor ?actor_country WHERE
                          { ?actor  dbpp:birthPlace  ?actor_country
                            FILTER regex(str(?actor_country), "USA") } }
                      { SELECT  ?movie ?genre WHERE
                          { OPTIONAL { ?movie  dbpo:genre  ?genre } } } } } } } } }
  \end{lstlisting}

\vspace*{5pt}
\pagebreak
\section{Naive \sparql Query for Topic Modeling}
\label{app:naiveQuery2}
\hfill

\begin{lstlisting}[
breaklines=true,
columns=fullflexible,
aboveskip=-1.0\baselineskip,
belowskip=-0.0\baselineskip,
language=SQL,
showspaces=false,
basicstyle=\ttfamily\scriptsize,
commentstyle=\color{gray},
keywordstyle=\color{blue},
otherkeywords={OPTIONAL, FILTER},
caption={Naive \sparql query corresponding to the \sparql query shown in Listing~\ref{lst:topicsparql}.},
captionpos=b,
label={lst:naivetopic}
]
PREFIX  swrc: <http://swrc.ontoware.org/ontology#>
PREFIX  rdf:  <http://www.w3.org/1999/02/22-rdf-syntax-ns#>
PREFIX  owl:  <http://www.w3.org/2002/07/owl#>
PREFIX  dcterm: <http://purl.org/dc/terms/>
PREFIX  xsd:  <http://www.w3.org/2001/XMLSchema#>
PREFIX  rdfs: <http://www.w3.org/2000/01/rdf-schema#>
PREFIX  foaf: <http://xmlns.com/foaf/0.1/>
PREFIX  dc:   <http://purl.org/dc/elements/1.1/>
PREFIX  dblprc: <http://dblp.l3s.de/d2r/resource/conferences/>

SELECT  ?title
FROM <http://dblp.l3s.de>
WHERE
  { 
    {SELECT ?paper WHERE {?paper rdf:type swrc:InProceedings}}.
    {SELECT ?paper ?author WHERE {?paper dc:creator ?author}}.       
    {SELECT ?paper ?date WHERE {?paper dcterm:issued  ?date }}.   
    {SELECT ?paper ?conference WHERE {?paper swrc:series ?conference}} .
    {SELECT ?paper ?title WHERE {?paper  dc:title  ?title}}.
    {SELECT ?paper ?date WHERE {?paper dcterm:issued  ?date FILTER ( year(xsd:dateTime(?date)) >= 2005 )  }}

    { SELECT ?author WHERE
      {
       { SELECT  ?author COUNT(?paper) as  ?count_paper
          WHERE
            {
              {SELECT ?paper WHERE {?paper rdf:type swrc:InProceedings}}.
              {SELECT ?paper ?author WHERE {?paper dc:creator ?author}}.       
              {SELECT ?paper ?date WHERE {?paper dcterm:issued  ?date }}.   
              {SELECT ?paper ?conference WHERE {?paper swrc:series ?conference}} .
              {SELECT ?paper ?title WHERE {?paper  dc:title  ?title}}.
              {SELECT ?paper ?date WHERE {?paper dcterm:issued  ?date FILTER ( year(xsd:dateTime(?date)) >= 2000 )  }} .
              {SELECT ?paper ?conference WHERE {?paper swrc:series ?conference FILTER( ?conference IN (dblprc:vldb, dblprc:sigmod) )}}
            }
          GROUP BY ?author
        }
        FILTER ( ?count_paper >= 20 ) 
      }
    }
  }
\end{lstlisting}

\end{appendix}
\end{document}